\documentclass[english,pra,amsfonts,amssymb,amsmath,twocolumn,groupedaddress]{revtex4-1}
\usepackage[utf8]{inputenc}
\pdfpageattr{/Group <</S /Transparency /I true /CS /DeviceRGB>>}
\usepackage{amsthm}
\usepackage{amsmath}
\usepackage{amssymb}
\usepackage{subcaption}

\PassOptionsToPackage{normalem}{ulem}
\usepackage{ulem}

\usepackage{comment}
\usepackage{caption}
\usepackage{algpseudocode}

\algnewcommand{\A}{\textbf{and}\space}
\algnewcommand{\Or}{\textbf{or}\space}
\algnewcommand{\Xor}{\textbf{xor}\space}
\makeatletter
\let\OldStatex\Statex
\renewcommand{\Statex}[1][3]{%
  \setlength\@tempdima{\algorithmicindent}%
  \OldStatex\hskip\dimexpr#1\@tempdima\relax}
\makeatother

\usepackage[unicode=true,pdfusetitle,
 bookmarks=true,bookmarksnumbered=false,bookmarksopen=false,
 breaklinks=false,pdfborder={0 0 1},backref=false,colorlinks=false]
 {hyperref}
\usepackage[usenames,dvipsnames]{xcolor}
\hypersetup{colorlinks=true, linkcolor=Maroon, citecolor=OliveGreen, filecolor=magenta, urlcolor=Blue}

\makeatletter

\theoremstyle{plain}
\newtheorem{thm}{\protect\theoremname}
\theoremstyle{plain}
\newtheorem{prop}[thm]{Proposition}
\newtheorem{lem}[thm]{Lemma}
\newtheorem{defin}[thm]{\protect\definname}
\newtheorem{observ}[thm]{\protect\observname}
\newtheorem{corol}[thm]{\protect\corolname}
\newtheorem{algorithm}[thm]{\protect\algorithmname}
\newtheorem{cnj}[thm]{Conjecture}
\newtheorem{example}[thm]{\protect\examplename}
\newtheorem{problem}[thm]{\protect\problemname}

\makeatother

\usepackage{babel}

\providecommand{\theoremname}{Theorem}
\providecommand{\definname}{Definition}
\providecommand{\observname}{Observation}
\providecommand{\corolname}{Corollary}
\providecommand{\algorithmname}{Algorithm}
\providecommand{\examplename}{Example}
\providecommand{\problemname}{Problem}

\usepackage{pgfplots}
\usetikzlibrary{calc}
\usepackage{pgfplotstable}
\usepackage{colortbl}
\usepackage{booktabs}
\usepackage{pgfplotstable}

\begin{document}
\global\long\def\Z{\mathbb{Z}}
\global\long\def\Q{\mathbb{Q}}
\global\long\def\R{\mathbb{R}}
\global\long\def\p{\varphi}
\global\long\def\t{\tau}
\global\long\def\th{\theta}
\global\long\def\w{\omega}
\global\long\def\ve{\varepsilon}
\global\long\def\st{\sqrt{\t}}
\global\long\def\sf{5^{\nicefrac{1}{4}}}
\global\long\def\sp{\sqrt{\p}}
\global\long\def\Zt{\Z\left[\tau\right]}
\global\long\def\T{\mathcal{T}}
\global\long\def\Id{\mathcal{I}}
\global\long\def\I{\mathcal{I}}
\global\long\def\F{\mathcal{F}}
\global\long\def\Ri{\Z\left[\w\right]}
\global\long\def\Zw{\Z\left[\w\right]}
\global\long\def\Ti{\mathbf{T}}
\global\long\def\Fi{\mathbf{F}}
\global\long\def\W{\mathcal{W}}
\global\long\def\m{\mathrm{\, mod\,}}
\global\long\def\k#1{\left|#1\right\rangle }
\global\long\def\gr#1{\left\langle #1\right\rangle }
\global\long\def\ip#1#2{\left\langle #1,#2\right\rangle }
\global\long\def\cp#1#2{\left[#1,#2\right]}
\global\long\def\nrm#1{\left\Vert #1\right\Vert }
\global\long\def\Ni#1{\left| #1\right|^2}
\global\long\def\Nt#1{N_\t\left(#1\right)}
\global\long\def\G#1{G\left( #1\right)}
\global\long\def\gc#1{#1^{\bullet}}
\global\long\def\cc#1{#1^{\ast}}
\global\long\def\cm#1{\Ni{\gc{#1}}}
\global\long\def\re{\mathrm{Re}}
\global\long\def\im{\mathrm{Im}}
\global\long\def\l{\left}
\global\long\def\r{\right}
\global\long\def\la{\langle}
\global\long\def\ra{\rangle}
\global\long\def\ym{y_{min}}
\global\long\def\yM{y_{max}}
\global\long\def\xm{x_{min}}
\global\long\def\xM{x_{max}}

\global\long\def\FT{\langle\F,\T\rangle}
\global\long\def\sgm{\langle\sigma_1,\sigma_2\rangle}

\global\long\def\set#1#2{\left\{  \left.#1\,\right|\,#2\right\}  }

\title{Asymptotically Optimal Topological Quantum Compiling}

\author{Vadym Kliuchnikov$^\dagger$}
\author{Alex Bocharov$^*$}
\author{Krysta M.~Svore$^*$}

\affiliation{$^\dagger$Institute for Quantum Computing and David R. Cheriton School of Computer Science, Univ.~of Waterloo, Waterloo, Ontario (Canada)\\
$^*$Quantum Architectures and Computation Group, Microsoft Research, Redmond, WA (USA)}

\begin{abstract}
In a topological quantum computer, universality is achieved by braiding and quantum information is natively protected from small local errors. 
We address the problem of compiling single-qubit quantum operations into braid representations for non-abelian quasiparticles described by the Fibonacci anyon model.
We develop a probabilistically polynomial algorithm that outputs a braid pattern to approximate a given single-qubit unitary to a desired precision. 
We also classify the single-qubit unitaries that can be implemented exactly by a Fibonacci anyon braid pattern and present an efficient algorithm to produce their braid patterns.
Our techniques produce braid patterns that meet the uniform asymptotic lower bound on the compiled circuit depth and thus are depth-optimal asymptotically.
Our compiled circuits are significantly shorter than those output by prior state-of-the-art methods, resulting in improvements in depth by factors ranging from 20 to 1000 for precisions ranging between $10^{-10}$ and $10^{-30}$.

\end{abstract}

\maketitle

\section{Introduction}
\label{sec:intro}

As hardware devices for performing quantum computation mature, the need for efficient quantum compilation methods has become apparent.
While conventional quantum devices will require vast amounts of error correction to combat decoherence, it is conjectured that certain quasiparticle excitations obeying non-abelian statistics, called non-abelian anyons, will require little to no error correction.
Certain quantum systems based on non-abelian anyons intrinsically protect against errors by storing information globally rather than locally and can be used for computation \cite{KitaevTop}.
If two quasiparticles are kept sufficiently far apart and their worldlines in $2+1$-dimensional space-time are braided adiabatically, a unitary evolution can be realized.
One class of non-abelian anyons, called Fibonacci anyons, are predicted to exist in systems in a state corresponding to the fractional quantum Hall (FQH) plateau at filling fraction $\mu = 12/5$~\cite{SankarFreedman,NayakFreedman}, where the theory is described by $SU(2)_k$ Chern-Simons-Witten theories \cite{ChernSimmons,Witten88,Witten89} for $k=3$.
In fact, it has been shown that for $k=3$ and $k>4$, $SU(2)_k$ anyons will realize universal quantum computation with braiding alone \cite{Freedman1,Freedman2}.

Previous work \cite{SimonFreedman,HormoziEtAl} has developed methods, using the Solovay-Kitaev algorithm \cite{DN}, for approximating a given single-qubit unitary to precision $\ve$ by a Fibonacci anyon braid pattern with depth $O(\log^c(1/\ve))$, where $c\sim 3.97$ in time $t\sim \log^{2.71}(1/\ve)$.
For coarse precisions, one can also use brute-force search to find a braid with minimal depth $O(\log(1/\ve))$ in exponential time \cite{HormoziEtAl}.
Since the number of braids grows exponentially with the depth of the braid, this technique is infeasible for long braids, which are required to achieve fine-grain precisions.
While the Solovay-Kitaev algorithm applies at any desired precision, including fine-grain precisions, it incurs a (costly) polynomial overhead that leads to compiled circuits of length $O(\log^{3+\delta}(1/\ve))$, far from the asymptotic lower bound of $O(\log(1/\ve))$.

In this work, we develop an algorithm for approximating a single-qubit unitary to precision $\ve$ by a Fibonacci anyon braid pattern with asymptotically optimal depth $O(\log(1/\ve))$.
Moreover, the algorithm requires only probabilistically polynomial runtime.
We also classify the set of unitaries that can be represented exactly ($\ve=0$) by a Fibonacci braid pattern and present an algorithm for their compilation.

A high-level representation of our compilation algorithm is depicted in Figure~\ref{fig:intro:algorithm:flow}.
The algorithm takes as input an arbitrary single-qubit unitary operation and a desired precision $\ve$ and approximates the unitary with a special unitary gate that can be represented exactly by a Fibonacci anyon braid pattern.
The exact synthesis algorithm is then applied to the special unitary gate and a $\FT$-circuit is synthesized.
A Fibonacci anyon braid pattern can always be written as an $\FT$-circuit and vice-versa, which we describe in detail in Section~\ref{sec:exact}.
Finally, we rewrite the circuit as a braid pattern and apply peephole optimization~\cite{pho}.

As will be shown in Section~\ref{sec:exact}, in order for a unitary matrix to be exactly representable as a Fibonacci anyon braid pattern, its entries must come from a certain ring of algebraic integers.
For that reason the approximation step is the most involved part of the algorithm.
It consists of two stages: in the first stage, we find an algebraic number that is in $\ve$-proximity to one of the elements of the input unitary matrix~(Section~\ref{sec:appr}).
This step can be viewed as a generalization of a numeric ``round-off" operation.
In the second stage, we complete the algebraic number with numbers from the same ring of algebraic integers such that all of the numbers form a unitary matrix.
This is done by solving the relative norm equation~(Section~\ref{sec:norm-equation}).

Solving the relative norm equation is in itself hard, and the task is further complicated by the fact that the equation is only solvable for a fraction of the generated ``round-offs";
testing if the given ``round-off" corresponds to a solvable norm equation is as hard as solving it.
However, there is a fraction of ``round-offs" that lead to easy instances of the problem and furthermore it is possible to efficiently test if the given instances are easy.

Without sacrificing too much quality in the approximation, we aim to find these easy instances of the problem by randomly generating ``round-offs" and testing if each instance is easy or not.
We prove, under a given number theory conjecture (Conjecture~\ref{cnj:appr:distr}), that the average number of iterations required to find an easy solvable instance is polynomial in the bit sizes of the inputs.
We also provide numerical evidence that supports the conjecture.
We carefully select the building blocks for the norm equation algorithm such that they also have probabilistically polynomial runtime for easy instances.

Our general approach to compilation into efficient single-qubit Fibonacci anyon circuits has been inspired in part by significant recent progress in efficient compilation into the $\langle H,T \rangle$ basis \cite{Selinger,KMM12,KMM1231}) and the $V$ basis \cite{BGS}. The unifying theme behind these advances has been the use of algebraic number theory, in particular the theory of cyclotomic fields as the foundation for algorithms designed to meet the asymptotic lower bounds on the complexity of compiled circuits.  Here, we employ a similar attack on circuit compilation through the use of algebraic number theory.

\begin{figure}[tb]
  \centering
\makeatletter
\pgfdeclareshape{datastore}{
  \inheritsavedanchors[from=rectangle]
  \inheritanchorborder[from=rectangle]
  \inheritanchor[from=rectangle]{center}
  \inheritanchor[from=rectangle]{base}
  \inheritanchor[from=rectangle]{north}
  \inheritanchor[from=rectangle]{north east}
  \inheritanchor[from=rectangle]{east}
  \inheritanchor[from=rectangle]{south east}
  \inheritanchor[from=rectangle]{south}
  \inheritanchor[from=rectangle]{south west}
  \inheritanchor[from=rectangle]{west}
  \inheritanchor[from=rectangle]{north west}
  \backgroundpath{
    \southwest \pgf@xa=\pgf@x \pgf@ya=\pgf@y
    \northeast \pgf@xb=\pgf@x \pgf@yb=\pgf@y
    \pgfpathmoveto{\pgfpoint{\pgf@xa}{\pgf@ya}}
    \pgfpathlineto{\pgfpoint{\pgf@xb}{\pgf@ya}}
    \pgfpathmoveto{\pgfpoint{\pgf@xa}{\pgf@yb}}
    \pgfpathlineto{\pgfpoint{\pgf@xb}{\pgf@yb}}
 }
}
\makeatother
\usetikzlibrary{arrows}
\begin{tikzpicture}[
  every matrix/.style={ampersand replacement=\&,column sep=1cm,row sep=0.3cm},
  sink/.style={draw,thick,rounded corners,fill=gray!20,inner sep=.2cm},
  datastore/.style={draw,very thick,shape=datastore,inner sep=.2cm},
  to/.style={->,>=stealth',shorten >=1pt,semithick},
  every node/.style={align=center}]

  \matrix{
    \& \node[datastore] (uni) {Unitary, precision $\ve$}; \& \\
    \& \node[sink] (appr) {Approximation algorithm}; \& \\
    \& \node[datastore] (euni) {Exact Unitary}; \& \\
 	\& \node[sink] (esa) {Exact synthesis algorithm}; \& \\
	\& \node[datastore] (ftc) {$\FT$-circuit}; \& \\
 	\& \node[sink] (copt) {Circuit optimization}; \& \\
	\& \node[datastore] (bp) {Fibonacci anyons braid pattern}; \& \\
  };
	\draw[to] (uni) --(appr);
	\draw[to] (appr) --(euni);
	\draw[to] (euni) --(esa);
	\draw[to] (esa) --(ftc);
	\draw[to] (ftc) --(copt);
	\draw[to] (copt) --(bp);
\end{tikzpicture}
  \caption[Compilation algorithm]{High-level flow of the compilation algorithm.
}
\label{fig:intro:algorithm:flow}
\end{figure}
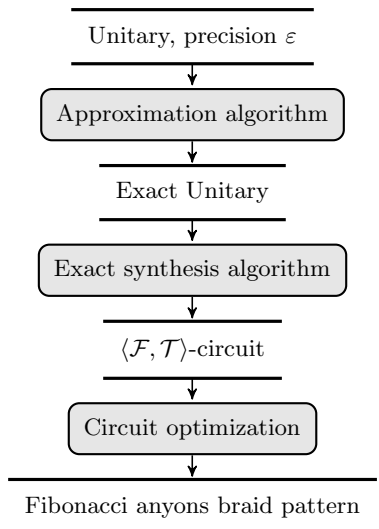

\section{Preliminaries}
\label{preliminaries:omega}

In this section, our main goal is to state the compilation problem and identify basic subproblems.  We also provide relevant detail of the mathematical and computational properties of small anyon ensembles.
We refer the reader to existing tutorials (such as \cite{TrebstTroyer,PreskillLN}) for detailed descriptions of the fundamental physics and mathematics of the anyon models.

\begin{figure}[tb]
  \begin{subfigure}[b]{\columnwidth}
    \caption{\label{fig:basis}Computational basis}
    \begin{tikzpicture}

\node[font=\Huge] at (0,1.5) {$|0\rangle : $};
\node[circle, fill = gray!20, draw = black, line width=1] (g1) at (2.0,1.5) {};
\node[circle, fill = gray!20, draw = black, line width=1] (g2) at (3.0,1.5) {};
\node[circle, fill = gray!20, draw = black, line width=1] (g3) at (4.0,1.5) {};
\node[font=\Huge] at (0,0) {$|1\rangle : $};
\node[circle, fill = gray!20, draw = black, line width=1] (g4) at (2.0,0) {};
\node[circle, fill = gray!20, draw = black, line width=1] (g5) at (3.0,0) {};
\node[circle, fill = gray!20, draw = black, line width=1] (g6) at (4.0,0) {};
\draw[line width=0.75] (2.5,0) ellipse (1cm and 0.4cm);
\node[font=\Large] at (3.6,-0.3) {1};
\node[font=\Large, color=black!70] at (4.9,-0.3) {1};
\draw[line width=0.75] (3,0) ellipse (1.75cm and 0.6cm);

\draw[line width=0.75] (2.5,1.5) ellipse (1cm and 0.4cm);
\node[font=\Large] at (3.6,1.2) {0};
\node[font=\Large, color=black!70] at (4.9,1.2) {1};
\draw[line width=0.75] (3,1.5) ellipse (1.75cm and 0.6cm);
\end{tikzpicture}
  \end{subfigure}%
  \vspace{0.5cm}
   \begin{subfigure}[b]{\columnwidth}
    \caption{\label{fig:braiding}Braiding}
    \input{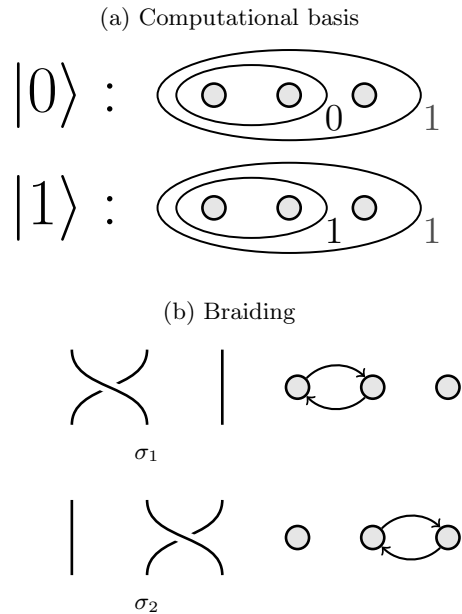}
  \end{subfigure}%
  \caption{Encoding a qubit with three Fibonacci anyons. Grey circles corresponds to particles, numbers near ellipses show particles' collective charge~\cite{HormoziEtAl}. }
\label{fig:anyons}
\end{figure}

There are several ways to simulate quantum circuits with Fibonacci anyon braid patterns.
For each such simulation, we need to define how qubits are encoded with anyons.
In this work, we focus on the three-particle encoding of a qubit~\cite{HormoziEtAl}, shown in Figure~\ref{fig:anyons}, where
the computational basis state $|0\rangle$ corresponds to the first two anyons having collective charge zero, and state $|1\rangle$ corresponds to the first two anyons having collective charge one.
The collective charge of all three anyons in both cases is one.
Measurement of the collective charge of the first two anyons is a projective measurement in the computational basis.
Unitary operations are realized by moving anyons around each other, where the result depends only on topological properties of anyon worldlines; the mathematical object that describes them is a braid group.
Figure~\ref{fig:braiding} shows anyon moves corresponding to the two generators $\sigma_1,\sigma_2$ of the three-strand braid group (corresponding to our the three particles).
We also use $\sigma_1,\sigma_2$ to denote the corresponding unitary operations, where
\[
 \sigma_1 = \w^{6} \l(\begin{array}{cc}
              1 & 0 \\
              0 & \w^7
            \end{array}\r), \w = e^{i\pi/5}.
\]
It is convenient to express $\sigma_2$ with a Fusion matrix $F$, which corresponds to one of the parameters defining the Fibonacci anyon model:
\[
F=\left(\begin{array}{cc}
\t & \st\\
\st & -\t
\end{array}\right), \sigma_2 = F \sigma_1 F.
\]
It has been shown \cite{Freedman1,Freedman2} that unitaries $\sigma_1,\sigma_2$ are approximately universal, e.g., for any single-qubit unitary and given precision $\ve$ there is a circuit consisting of $\sigma_1$ and $\sigma_2$ gates which approximates U within precision $\ve$.

In this paper we describe an algorithm for approximating an arbitrary single-qubit unitary using circuits over a gate set consisting of $\sigma_1$, $\sigma_2$, $\sigma^{-1}_1$, and $\sigma^{-1}_2$.
We call such a circuit a $\sgm$-circuit and refer to the four basis gates as $\sigma$ gates.
The circuit length (or equivalently depth) corresponds to the number of anyon moves needed to perform the given unitary.
It defines how efficiently a given circuit implements a unitary.

More formally, the problem we solve can be stated as:
\begin{problem} \label{prelim:compilation:problem}
Let $U \in U(2)$ be a single-qubit unitary and $\ve>0$ an arbitrary positive number.
Find a $\sgm$-circuit $c$ of length at most $O(\log(1/\ve))$ such that the corresponding unitary is within distance $\ve$ from $U$.
\end{problem}
We use a global phase-invariant distance to measure the quality of approximation
\begin{equation*}
d(U,V)=\sqrt{1-\l|tr(UV^{\dagger})\r|/2}.
\end{equation*}

Problem \ref{prelim:compilation:problem} can reduced to approximating two types of unitaries:
rotations around the $Z$ axis
\[
R_z(\phi) := \l(\begin{array}{cc} e^{-i\phi/2} & 0 \\ 0 & e^{i\phi/2}          \end{array}\r)
\]
and $R_z(\phi)X$, where $X$ is the Pauli $X$ gate.
These two problems are in fact different because the $X$ gate can only be approximated with a Fibonacci anyons braid pattern.
The next lemma shows how the problem of approximating an arbitrary unitary can be reduced to these two types:

\begin{lem} \label{prelim:decomposition:lemma}
Consider a single-qubit unitary $U \in U(2)$. If the upper left entry of $U$ is non-zero then it can be represented as
\begin{equation}
 e^{i \delta} \, R_z(\alpha)\, \F \, R_z(\beta) \, \F \, R_z(\gamma), \alpha, \beta, \gamma, \delta \in \R,
\end{equation}
otherwise it can be represented as $e^{i \delta}R_z(\phi)X$ for $\phi,\delta\in\R$.
\end{lem}
\begin{proof}
First, we factor out an appropriate global phase from $U$ to obtain a special unitary $U'$. The product $R_Z(\alpha)\, \F \, R_Z(\beta) \, \F \, R_Z(\gamma)$ corresponds to a special unitary with upper left entry
\begin{equation*}
e^{-i\,(\alpha+\beta+\gamma)/2} \, \t \, (e^{i \, \beta}+\t)
\end{equation*}
which is always non-zero.
It is not difficult to solve the equality  $R_Z(\alpha)\, \F \, R_Z(\beta) \, \F \, R_Z(\gamma)=U'$ when $U'$ has a non-zero upper left entry; in the other case it is not difficult to find $\phi$ such that $U'=R_z(\phi)X$.
\end{proof}

\section{Exact synthesis algorithm} \label{sec:exact}
The goal of this section is two-fold: first, to classify all single-qubit unitaries that can be implemented exactly by braiding Fibonacci anyons and second, to describe an efficient algorithm for finding a circuit that corresponds to a given exactly implementable unitary. 
We start by introducing rings of integers and use them to describe the most general form of exactly implementable unitaries. 
Next, we introduce a complexity measure over the unitaries by using ring automorphisms. 
Finally, we describe the exact synthesis algorithm: a process, guided by the complexity measure, to find a circuit for an exactly synthesizable unitary.

For the purpose of this section it is more convenient to consider the elementary gates
\[
\T=\left(\begin{array}{cc}
1 & 0\\
0 & \w
\end{array}\right),\F=\left(\begin{array}{cc}
\t & \st\\
\st & -\t
\end{array}\right)
\]
instead of $\sigma_1$ and $\sigma_2$. We call a circuit composed of $\F,\T$ and $\w\Id$ gates an $\FT$-circuit, where $\Id$ is the identity gate and $\w\Id$ is a global phase. 
The following relations imply that any unitary that is implementable by an $\FT$-circuit is also implementable by a $\sgm$-circuit and vice-versa:
\begin{equation}
\begin{array}{lcl}
\sigma_1 = (\w\Id)^{6}\T^7,&\quad& \T = (\w\Id)^{2}(\sigma_1)^3,\\
\sigma_2 = (\w\Id)^{6}\F\T^7\F,&\quad& \F = (\w\Id)^{4}\sigma_1\sigma_2\sigma_1.
\end{array}
\end{equation}
For the applications considered in this paper the global phase $\w\Id$ is irrelevant.

Two rings are crucial for understanding our results. The first is the ring of integers extended by the primitive tenth root of unity $\w$
\[
 \Zw := \set{a+b\w+c\w^2+d\w^3}{a,b,c,d\in\Z}
\]
and the second is its real subring
\[
 \Zt := \set{a+b\t}{a,b\in\Z}.
\]
It is not difficult to check that both definitions are correct: the sets defined above are both rings. It is straightforward to check that both sets $\Zw$ and $\Zt$ are closed under addition. To show that both sets are closed under multiplication it is sufficient to check the following equalities:
\begin{equation}
\begin{array}{ccc}
\w^4 = -1+\w-\w^2+\w^3, & \w^5 = -1, & \t^2 = 1-\t.
\end{array}
\end{equation}
Equality $\t=\w^2-\w^3$ implies that $\Zt$ is a subring of $\Zw$. Both rings are well studied in algebraic number theory. 

In this section we use elementary methods that do not require any background. We discuss how objects and results of the section are related to the broader mathematical picture in Appendix \ref{sec:exact:math}. For example, we show why $\Zt$ is a real subring of $\Zw$ (i.e., equal to $\Zw\cap\R$). For the purpose of this section it is sufficient to note that for any $u$ from $\Zw$ its absolute value squared $|u|^2$ is from $\Zt$.

An \emph{exact} unitary is defined by an integer $k$ and $u,v$ from $\Zw$ such that $|u|^2+\t|v|^2=1$ and
\begin{equation}
\label{exact:reresentable:unitary}
U[u,v,k]:=\left(\begin{array}{cc}
u & v^*\,\sqrt{\t}\,\w^k\\
v \, \sqrt{\t} & -u^*\,\,\w^k
\end{array}\right).
\end{equation}
The following Lemma explains the intuition behind the name.
\begin{lem}[Exact synthesis lemma]
A unitary in $U(2)$ can be expressed as a product of $\F$ and $\T$ matrices if and only if it is exact.\end{lem}
\begin{proof}
The ``only if" part of the lemma is straightforward. Since $\t \in \Zw$,  both generators $\T$ and $\F$ have form (\ref{exact:reresentable:unitary}) and multiplying a matrix of this form by either $\T$ or $\F$ preserves the form.

We defer the proof of the converse until we fully develop the exact synthesis algorithm (Figure~\ref{fig:exact:synthesis}) and prove its correctness. The ``if" part of the lemma follows from that correctness proof (see Theorem~\ref{thm:exact:correctness}).
\end{proof}

An automorphism of $\Zw$ that plays a fundamental role in the construction of the exact synthesis algorithm is the mapping
\begin{equation}\label{exact:synthesis:bullet}
\gc{(\ldotp)}:\Zw\mapsto\Zw \text{ such that } \gc{\w}=\w^3.
\end{equation}
By definition, a ring automorphism must preserve the sum and the product. For example, we find that
\[
 \gc{\t} = \gc{(\w^2-\w^3)}=\gc{(\w^2)}-\gc{(\w^3)}=(\gc{\w})^2-(\gc{\w})^3=-\p.
\]
Taking into account that $\p=\t+1$~(e.g., $\p$ is from $\Zt$) we see that $\gc{(\ldotp)}$ restricted on $\Zt$ is also an automorphism of $\Zt$. The other important property of automorphism $\gc{(\ldotp)}$ is its relation to complex conjugation (which is another automorphism of $\Zw$):
\begin{equation}
  \gc{(\gc{x})} = \cc{x}.
\end{equation}
For example, this implies the equality $\gc{(\cc{x})}=\cc{(\gc{x})}$, which is used in several proofs in this work. It also implies that $\cm{x}=\gc{(\Ni{x})}$ because $\Ni{x}=x\cc{x}$.

The Gauss complexity measure $G$ \cite{HLenstraJr} is a key ingredient of the exact synthesis algorithm
\begin{equation}
\G{u}:= \Ni{u} + \cm{u}.
\end{equation}
We also extend the definition to exact unitaries as $\G{U[u,v,k]} = \G{u}$. Roughly speaking, the exact synthesis algorithm~(Figure~\ref{fig:exact:synthesis}) reduces the complexity of the unitary $U_r$ by multiplying it by $\F\T^k$, while $\G{U_r}$ describes how complex $U_r$ is at each step. The example in Figure~\ref{fig:exact:example} illustrates the intuition behind $G$. The following proposition summarizes important properties of the Gauss complexity measure.
\begin{figure}[t]
\begin{tabular}{|c|c|c|}
\hline
$n$ & $u_{n}$ & $\G{u_n}$ \tabularnewline \hline
\hline
$0$ & $1$                & $1$  \tabularnewline \hline
$1$ & $\w^{2}-\w^{3}$    & $3$  \tabularnewline \hline
$2$ & $2-\w+\w^{3}$      & $13$ \tabularnewline \hline
$3$ & $-3+5\w-2\w^{2}-1$ & $57$ \tabularnewline \hline
\end{tabular}
\caption{\label{fig:exact:example}Values of the Gauss complexity measure for the family of unitaries $U[u_n,v_n,k_n]=(\F\T)^n $, $n$ from $\{0,1,2,3\}$.}
\end{figure}
\begin{prop}\label{prop:exact:alternatives}
For any $x$ from $\Zw$, the Gauss complexity measure $\G{x}$ is an integer, and there are three alternatives:
\[ \begin{array}{cl}
    (a) & \G{x} = 0,\text{ then } x = 0, \\
    (b) & \G{x} = 2,\text{ then } x = \w^k \text{ for } k \text{ -- integer}, \\
    (c) & \G{x} \ge 3,
   \end{array} \]
and $\G{a+b\w+c\w^2+d\w^3} \le 5/2\l(a^2+b^2+c^2+d^2\r)$.
\end{prop}
\begin{proof}
As observed in \cite{HLenstraJr}, $\G{a+b\w+c\w^2+d\w^3}$ is a positive definite quadratic form when viewed as a function of $a,b,c,d$. We found by direct computation that it takes integer values when $a,b,c,d$ are integers, its minimal eigenvalue is $1/2$ and its maximal eigenvalue is $5/2$. This implies an upper bound on $G$ in terms of $a,b,c,d$ and an inequality
\[
 \G{a+b\w+c\w^2+d\w^3} \ge (a^2+b^2+c^2+d^2)/2.
\]
Implication in $(a)$ follows immediately from above. The inequality also allows us to prove the proposition by considering a small set of special cases
\[
 (a,b,c,d) \in S = \{-2,-1,0,1,2\}^4.
\]
In all other cases, $\G{a+b\w+c\w^2+d\w^3}$ is greater than four, which corresponds to alternative $(c)$. To finish the proof it is sufficient to exclude quadruples corresponding to $0,\w^k$ from $S$ and find that the minimum of  $\G{a+b\w+c\w^2+d\w^3}$ over the remaining set is $3$.
\end{proof}

\begin{figure}[t]
\begin{algorithmic}[1]
\Require $U$ -- exact unitary (i.e., in the form of (\ref{exact:reresentable:unitary}))
\Procedure{EXACT-SYNTHESIZE}{$U$}
\State $g \gets \G{U}, U_r\gets U, C \gets (\text{empty circuit})$
\While{ $g>2$ }
  \State $J \gets {\arg\min}_{j\in\{1,\ldots,10\}} \G{\F\T^j U_r}$
  \State $U_r \gets \F\T^{J}U_r, g \gets \G{U_r}$
  \State Add $\F\T^{10-J}$ to the beginning of circuit $C$
\EndWhile
\State Find $k,j$ such that $U_r=\w^k \T^j$
\State Add $\w^k\T^{j}$ to the beginning of circuit $C$
\EndProcedure
\Ensure $C$ -- circuit over $\langle\F,\T\rangle$ that implements $U$
\end{algorithmic}
\caption{\label{fig:exact:synthesis}Exact synthesis algorithm.}
\end{figure}

To prove the correctness of the exact synthesis algorithm (Figure~\ref{fig:exact:synthesis}) we need the following technical result.

\begin{lem}\label{lem:exact:ratio}
For any $u,v$ from $\Zw$ such that $|u|\ne1$, $|u|\ne0$, $\Ni{u}+\t\Ni{v}=1$,
there exists $k_0(u,v)$ such that:
\[
\begin{array}{cl}
  (a) & \G{(u+\w^{k_0(u,v)}v)\t}/\G{u} < 1, \\
  (b) & \G{(u+\w^{k_0(u,v)}v)\t}/\G{u} < \frac{\p^2}{2}(\sqrt[4]{5}-1)^2 + \\
      & r(\G{u}), \text{ where }r(n) \text{ is in }O(1/n)\\
\end{array}
\]
In addition, for any $k$, the ratio $\G{(u+\w^{k}v)\t}/\G{u}$ is upper bounded by $ \frac{3\p^2}{2}(1 + \sqrt{\t})^2$.
\end{lem}

We first prove that the exact synthesis algorithm is correct and efficient and then proceed to the proof of Lemma \ref{lem:exact:ratio}.

\begin{thm} \label{thm:exact:correctness}
For any exact unitary $U$ (in the form of (\ref{exact:reresentable:unitary})):
\begin{enumerate}
\item the exact synthesis algorithm (Figure~\ref{fig:exact:synthesis})
terminates and produces a circuit that implements $U$,
\item $n$, the minimal length of $\FT$-circuit implementing $U$, is in $\Theta(\log(\G{U}))$,
\item the algorithm requires at most $O(n)$ arithmetic operations and outputs an $\FT$-circuit with $O(n)$ gates
\end{enumerate}
\end{thm}
\begin{proof}
The termination of the while loop in the algorithm follows from two facts: the complexity measure of $U$ is an integer and, by Lemma \ref{lem:exact:ratio}, it strictly decreases at each step. Indeed, consider ratio $\G{\F\T^j U_r} / \G{U_r}$. If we denote the upper left entry of $U_r$ by $u$ and the lower left by $v\sqrt{\t}$, then the ratio is precisely equal to the one considered in Lemma~\ref{lem:exact:ratio}. By picking $j$ that minimizes $\G{\F\T^j U_r}$ we ensure that implications $(a)$ and $(b)$ of the Lemma hold. After the loop execution we have $\G{U_r}\le2$. According to Proposition \ref{prop:exact:alternatives}, the only possible values of $G$ of the upper left entry $u$ of $U$ are either $0$ or $2$. There is no exactly synthesizable unitary with $u=0$. In other words, there is no $v$ from $\Zw$ such that equation $\tau|v|^2=1$ is solvable (see Section~\ref{sec:norm-equation}). The only remaining case is $\G{u}=2$. By Proposition~\ref{prop:exact:alternatives}, $u$ must be a power of $\w$,
therefore $U_r$
must be representable as $\w^k\T^j$. Correctness of the algorithm follows from the fact that during the algorithm execution at steps $3$ and $6$, it is always true that $U=U_C U_r,$ where $U_C$ denotes a unitary that is implemented by circuit $C$. This implies that any exactly synthesizable unitary $U$ can be represented as a $\FT$-circuit.

Now we prove the second and the third statements. Taking into account that $\F^2=\Id$ and $\T^{10}=\Id$, any exact unitary can be represented as the following matrix product
\begin{equation}\label{exact:product-representation}
 \w^{k(m+1)}\T^{k(m)}\F\T^{k(m-1)}\F\T^{k(m-2)}\ldots\F\T^{k(0)}
\end{equation}
where $k(j)$ are from $\{0,\ldots,9\}$. The exact synthesis algorithm produces the circuit that leads precisely to representation~(\ref{exact:product-representation}) and $m$ in this case is the number of steps performed by the algorithm. Lemma~\ref{lem:exact:ratio} implies that $m$ is upper bounded by $\log_3(\G{U})+c_1$, as $\frac{\p^2}{2}(\sqrt[4]{5}-1)^2 < 1/3$. This gives an upper bound on the length of the circuit produced by the algorithm and also on $n$, the minimal possible length of any circuit that implements $U$.

Consider now representation~(\ref{exact:product-representation}) obtained from a minimal length $\FT$-circuit implementing $U$. We prove a bound on $\G{U}$ by induction. First, by direct computation, $\G{\F\T^{k(0)}}$ is equal to three. Next, we introduce $V_j=\F\T^{k(j-1)}\ldots\F\T^{k(0)}$ and by the third part of Lemma~\ref{lem:exact:ratio} find
\[
 \G{\F\T^{k(j)} V_j} \le \frac{3\p^2}{2}(1 + \sqrt{\t})^2 \G{V_j}.
\]
We note that multiplication by $\w^{k(m+1)}\T^{k(m)}$ does not change the complexity measure and conclude that $\log(\G{U})$ is upper bounded by some linear function of $m$. It is not difficult to see that $m$ is less than $n$. This finishes the proof of the theorem.
\end{proof}
In Lemma \ref{lem:exact:ratio} we show that the question about the change of the complexity measure is rather geometrical. On the high level, the relative change of complexity measure is related to an angle between certain vectors and the power of $\w$ in the expression $G(\t(u+\w^k v))$ allows to adjust the angle by multiples of $\pi/5$.
\begin{proof}[Proof of Lemma \ref{lem:exact:ratio}]
We first prove inequality $(a)$. Consider the special case when there exists $k$ such that $\G{(u+\w^k v)\t}=2$. We define $k_0(u,v)=k$. Inequality $(a)$ holds because $|u|\ne 1$ and $|u|\ne0$ implies $\G{u}\ge3$ by Proposition \ref{prop:exact:alternatives}.

Now we assume that $\G{(u+\w^k v)\t}\ge3$. It is more convenient to consider ratio $\cm{((u+\w^k v)\t)}/\cm{u}$ instead of $\G{(u+\w^{k_0(u,v)}v)\t}/\G{u}$. The equality $\Ni{\gc{u}}=\G{u}-\Ni{u}$, inequalities $|u|<1$, $\G{(u+\w^k v)\t}\ge3$ and multiplicative property of $\gc{(\ldotp)}$ imply
\[
 \frac{\G{(u+\w^k v)\t}}{\G{u}} < \frac{3\p^2}{2} \cdot \frac{ \cm{(u+\w^k v)}}{\cm{u}}.
\]
Next we expand $\cm{(u+\w^k v)}$ and see that it is equal to
\[
 \cm{u} + \cm{v} +2\re(\gc{u}\cc{(\gc{v})}\w^{-3k}).
\]
Here we used that $\gc{\w}=\w^3$. It is convenient to introduce
\[
e^{i\phi} := \frac{\gc{u}}{|\gc{u}|} \cc{\l( \frac{\gc{v}}{|\gc{v}|} \r)} \,\alpha := \frac{\cm{v}}{\cm{u}}.
\]
In terms of $\alpha,e^{i\phi}$ we find
\begin{equation}\label{eq:ratio}
  \frac{\cm{(u+\w^k v)}}{\cm{u}} = 1 + \alpha +2\re\l(e^{i(\phi-3\pi k/5)}\r)\sqrt{\alpha}.
\end{equation}
It is always possible to chose $k_0(u,v)$ such that
\begin{equation}\label{ineq:subratio}
 \frac{ \cm{(u+\w^{k_0(u,v)} v)}}{\cm{u}} < 1+\alpha - 2\sqrt{\alpha}\cos(\pi/10).
\end{equation}
Using the estimate above we get
\begin{equation}\label{ineq:ratio}
 \frac{\G{(u+\w^{k_0(u,v)} v)\t}}{\G{u}} < \frac{3\p^2}{2} \cdot\l( 1+\alpha - 2\sqrt{\alpha}\cos(\pi/10) \r).
\end{equation}
Note that $\alpha$ is a function of $\cm{u}$. Indeed, $\Ni{u}+\t\Ni{v}=1$ implies
\[
 1 = \gc{(\Ni{u}+\t\Ni{v})} = \cm{u} - \p\cm{v}.
\]
We conclude that $\alpha=\t(\cm{u}-1)/\cm{u}$. Ratio ${\G{(u+\w^{k_0(u,v)} v)\t}}/{\G{u}}$ is upper bounded by the value of the right hand side of inequality (\ref{ineq:ratio}) at $\cm{u}=2$ which is approximately $0.9982$. Indeed, the right hand side is a monotonically decreasing function of $\cm{u}$. In addition, it is always the case that $\cm{u}>2$.
This follows from the relation between $\cm{u}$ and $\G{u}$ and the inequality $\G{u}\ge3$. This completes the proof of $(a)$.

To show $(b)$ we first note that
\[
 \frac{\G{(u+\w^k v)\t}}{\G{u}} < \frac{\G{(u+\w^k v)\t}}{\cm{((u+\w^k v)\t)}} \cdot \frac{ \p^2 \cm{(u+\w^k v)}}{\cm{u}}.
\]
To complete the proof of $(b)$ it is sufficient show that the first ratio in the right hand side of the inequality above is always less than $1+r_1(\G{u})$ and the second one is less than $\frac{\p^2}{2}(\sqrt[4]{5}-1)^2 + r_2(\G{u})$ when $k=k_0(u,v)$, for $r_{1,2}(n)$ from $O(1/n)$. The definition of $\G{\ldotp}$ implies
\[
 \frac{\G{(u+\w^k v)\t}}{\cm{(u+\w^k v)}} < 1 + 1/\cm{(u+\w^k v)}.
\]
It follows from equation~(\ref{eq:ratio}) that there exists $C_1$ such that $1/\cm{(u+\w^k v)} < C_1/\cm{u}$. We conclude that $r_1(x)=C_1/(x-1)$. Next we rewrite the right hand side of equation~(\ref{ineq:subratio}):
\[
 \frac{ \cm{(u+\w^{k_0(u,v)} v)}}{\cm{u}} < 1 + \t + 2\sqrt{\t}\cos(\pi/10) + f(\cm{u}).
\]
Taking into account that $f$ is monotonically decreasing and $\cm{u}>\G{u}-1$, we define $r_2(x)=f(x-1)$. By direct computation, we note that
\[
 1 + \t + 2\sqrt{\t}\cos(\pi/10) = \frac{1}{2}(\sqrt[4]{5}-1)^2
\]
which completes the proof of $(b)$.

Now we obtain the bound for arbitrary $k$. We use equation~(\ref{eq:ratio}) and inequality $\alpha < \t$ to find
\[
   {\cm{(u+\w^k v)}}/{\cm{u}} \le (1 + \sqrt{\t})^2.
\]
As before, by distinguishing cases when $\G{(u+\w^k v)\t}$ is equal to two or greater than two we get
\[
   \frac{\G{(u+\w^k v)\t}}{\G{u}} < \frac{3\p^2}{2}(1 + \sqrt{\t})^2
\]
which finishes the proof.
\end{proof}


\section{norm equation}
\label{sec:norm-equation}

In this section, we begin by explaining the role of norm equations in the approximate synthesis of Fibonacci anyon circuits and proceed with a detailed description of solvability conditions and the essential procedures required for solving norm equations. In the concluding subsection, we present an algorithm for solving a certain class of norm equations in probabilistically polynomial runtime.

\subsection{Motivation}
As outlined in Section \ref{sec:intro}, a method for completing an element from the ring $\Ri$ by other elements from the ring is required, such that together they make up a unitary matrix of the form (\ref{exact:reresentable:unitary}).

For example, consider first the special case where our compilation target is the Pauli $X$ gate
\[
X=\l(\begin{array}{cc}
              0 & 1 \\
             1  & 0
            \end{array}\r) .
\]

As per (\ref{exact:reresentable:unitary}), the off-diagonal element $1$ needs to be represented as $1=v\,\sqrt{\t}, v \in \mathbb{C}$, which makes $v$ equal to $\sqrt{\t^{-1}}=\sqrt{\phi}$. Since the latter does not belong to $\Ri$ it must be approximated by an element of $\Ri$.
Suppose we have found a $u \in \Ri$ such that $|u - \sqrt{\phi}| < \epsilon$. It can be shown that $|u|$ cannot be made exactly equal to $\sqrt{\phi}$, therefore the matrix
\[
\l(\begin{array}{cc}
              0 & -u^*\,\sqrt{\t}  \\
             u\,\sqrt{\t}  & 0
            \end{array}\r)
\]
is going to be subtly non-unitary, thus we must fill the diagonal with elements $v,v^* \in \Ri$, however tiny, such that $|v|^2+|u|^2 \, \t = 1$.
This amounts to solving the equation $|v|^2=1-|u|^2 \, \t$ for $v \in \Ri$.

Now consider a more general case, recalling Lemma \ref{prelim:decomposition:lemma}, where the compilation target is a $Z$-rotation
\[
R_Z(\theta)=\l(\begin{array}{cc}
              e^{-i \, \theta/2} & 0 \\
             0  & e^{i \, \theta/2}
            \end{array}\r).
\]
Suppose $u \in \Ri$ and $|u - e^{-i \, \theta/2}| < \epsilon$.
In general, $|u|$ will be close to $1$ but not exactly $1$, so
\[
\l(\begin{array}{cc}
              u & 0 \\
             0  & u^*
            \end{array}\r)
\]
is likely to be subtly non-unitary.
Again, we need to fill the off-diagonal entries with $v \, \sqrt{\t}; -v^* \, \sqrt{\t}, v \in \Ri$, which amounts to solving $|v|^2 = (1-|u|^2)/\t$ for $v \in \Ri$.

In order to develop a general solution for solving these norm equations, we must consider the previously defined rings $\Zt$ and $\Ri$ and automorphism
$\gc{(\ldotp)}$.

\subsection{Components of the norm equation}

Consider the following \textit{norm maps}, further described in Appendix \ref{sec:exact:math}:
\begin{eqnarray}
N_i : \Ri &\rightarrow& \Zt ; \, N_i(\eta) = \eta \, \eta^*, \\\label{eq:relative:norm}
N_{\t}:\Zt &\rightarrow& \Z ; \, N_{\t}(\xi) = \xi \, \xi^{\bullet},\\\label{eq:tau-norm}
N: \Ri &\rightarrow& \Z ; \, N(\eta) = N_{\t}N_i(\eta),\label{eq:absolute:norm}
\end{eqnarray}
where map (\ref{eq:absolute:norm}) is called the \textit{absolute norm} for the ring $\Ri$ and
map (\ref{eq:relative:norm}) is called the \textit{relative norm} in the ring extension $\Ri/\Zt$ and is shorthand for the squared complex absolute value: $N_i(x)=|x|^2$.
Correctness of its interpretation as a map into the ring $\Zt$ stems from the fact that $\Zt$ is the largest real subring of $\Ri$ or, in other words, $\Zt = \Ri \cap \mathbb{R}$ (see Appendix \ref{sec:exact:math}).
Therefore the real-valued element $\eta \, \eta^* \in \Zt$.


Given a $\xi\in\Zt$, we introduce the relative norm equation that we want to solve for $x\in\Ri$:
\begin{equation}
N_{i}(x)=\xi.\label{eq:rnorm}
\end{equation}
The two necessary conditions for (\ref{eq:rnorm}) to be solvable are $\xi>0$ and $\xi^{\bullet}>0$, but these conditions are in general not sufficient.
We develop the complete set of conditions in the next two subsections.

\subsection{Units and the greatest common divisor}

For ring extensions such as $\Ri/\Zt$, the theory
of norm equations is relatively simple, and the solutions
of such equations are well understood.
It is particularly simple here due to the fact that both $\Zt$ and $\Ri$ are principal ideal domains (PIDs) \cite{Jacobson}.
We describe the PID property and its consequences, and refer the reader to Appendix \ref{sec:exact:math} for details.

First, we note that given a ring $R$ and a principal ideal generated for an element $p \in R$ , i.e., $I(p)=\{ r\, p | r \in R \}$, the representation of the ideal is unique up to an invertible element of $R$, called a \textit{unit}.
If $u \in R$ is a unit, any element $ r \, p \in I(p) $ is equal to $ (r \, u^{-1}) \, (u \, p) $ and thus belongs to $I(u\,p)$.
It follows that the converse is also true.

For example, when $R=\mathbb{Z}$ the only units are $+1$ and $-1$.
The ideal $I(2)$ contains both positive and negative even numbers and thus coincides with $I(-2)$.
However, the group of units of $\Zt$ is infinite (see Lemma \ref{UNITS:LEMMA}) and therefore each ideal in $\Zt$ has an infinite set of equivalent representations.

Second, when $R$ is a principal ideal domain any two elements of the ring have at least one greatest common divisor.
For $a,b \in R$, the ideal $I(a,b) = \{ v\, a + w \, b | v,w \in R\}$ must be generated by some $g \in R$ and any common divisor of $a$ and $b$ divides $g$.
Again, a greatest common divisor is only unique up to multiplicative unit of the ring $R$.
Units of a number ring $R$ play an important role in several of the algorithms below.

The following definition will be useful:
\begin{defin}
(1) The set of of all units of a ring $R$ is a group with respect to multiplication called the \textbf{unit group} of $R$ and denoted by $U(R)$.

(2) Two elements $r_1, r_2 \in R$ are \textbf{associate} if there exists a unit $u \in U(R)$ such that $r_2 = u \, r_1$.
\end{defin}

\begin{example} \label{associate:example}
(1) As per lemma \ref{UNITS:LEMMA}, the unit group $U(\Zt)$ is an infinite group generated by $\{-1,\t\}$ and it consist of all elements $\{\pm\t^k, \, k \in \Z\}$.

(2) By direct computation in $\Zt$, $\, (2-\t)^{\bullet} = (3+\t) = \t^{-2}\,(2-\t)$. Thus $(2-\t)^{\bullet}$ is associate with $(2-\t)$.
\end{example}

Both rings  $\Zt$ and $\Ri$ are principal ideal domains; in Section \ref{subsection:GCD:Zomega} an efficient algorithm for computing the greatest common divisor (GCD) in $\Ri$ is presented.
Below, $GCD_{\Ri}$ refers to a greatest common divisor computed by this method or otherwise.

The goal of the following subsections is to develop an algorithm for solving the norm equation (\ref{eq:rnorm}) and to prove that its runtime is probabilistically polynomial in the bit size of the righthand side of (\ref{eq:rnorm}) provided a prime factorization of the righthand side is available.

\subsection{Solvability and solutions}

First we identify the \textit{primitive element} $\theta \in \Ri$:
\begin{equation}
\theta=\w+\w^{4}=-1+2\,\w-\w^{2}+\w^{3},\label{eq:primitive}
\end{equation}
where by direct computation
\begin{equation}
\theta^2=\t-2\label{eq:thetaSquare}.
\end{equation}
Since $\t-2 < 0$, the value of $\theta$ in $\mathbb{C}$ is purely imaginary.



We begin by investigating Eq (\ref{eq:rnorm}) for the special case where $\xi$ is prime.
First, consider the case when the norm of $\xi$ is 5:
\begin{observ} \label{normeq:norm5}
If $\xi>0, \xi^{\bullet}>0$ and $N_{\t}(\xi)= 5$, there exists a $k\in\Z$, such that $x=\pm \t^{k}\,\theta$ is a solution of (\ref{eq:rnorm}), where $\theta$ is the primitive element (\ref{eq:primitive}).
\end{observ}

\begin{proof}
Note first that $N_i(\theta)=2-\t$ and $N_{\t}(2-\t)=5$.
As follows from Appendix \ref{sec:app:B}, any other solution of $N_{\t}(\xi)=5$ is associate with either $(2-\t)$ or $(2-\t)^{\bullet}$.
However, the $(2-\t)$ element is exceptional in $\Zt$. It is associate with its adjoint (see Example \ref{associate:example}). Therefore $\xi = u \, (2-\t)$ where $u \in \Zt$ is a unit. Units of $\Zt$ are characterized in Lemma \ref{UNITS:LEMMA}. For $\xi>0, \xi^{\bullet}>0$ to hold, $u$ must be of the form $u = \t^{2k}, k \in \Z$, and our observation immediately follows.
\end{proof}

A general solution is based on the following:
\begin{thm}\label{solve:prime:xi}
Given a prime element $\xi\in\Zt$, the norm equation $N_{i}(x)=\xi$
is solvable in $\Ri$ if and only if the following two conditions
are satisfied:
\begin{enumerate}
\item $\xi>0,\xi^{\bullet}>0$,
\item $p=N_{\t}(\xi)$ is an integer prime that is either of the form $p=5\,m+1, m \in \Z$ or $p=5$.
\end{enumerate}
Assuming these conditions are satisfied, then

(a) $\t-2=m^{2}\m\xi$ for some $m\in\Z$,

(b) when $p\neq\pm5$, the equation (\ref{eq:rnorm}) has at least two distinct solutions $x=\tau^{k}\, s$ and $x^{*}=\tau^{k}\, s^{*}$ for a certain $k\in\Z$,
where $s=GCD_{\Ri}(\xi,m-\theta)$.
\end{thm}
The proof requires extensive algebraic number theory and is outlined in Appendix \ref{sec:app:B}.

Powerful general algorithms exist for solving relative norm
equations in algebraic field extensions \cite{HCohen99,DSimon,UGPARI}
, however Thm \ref{solve:prime:xi} suggests an algorithm that has probabilistically polynomial time complexity for important cases.
The following methods are required:
\begin{enumerate}
\item an algorithm for computing square root modulo a prime,
\item an algorithm for computing GCD in $\Ri$,
\item an algorithm for computing $\log_{\t}(unit)$ (i.e., the one that recovers the integer $m$ given the value of $\t^m$).
\end{enumerate}

Before proceeding, we first present a general theorem describing solutions of (\ref{eq:rnorm}) where the righthand side is not necessarily prime.

Since $\Zt$  is a PID, we can factor any element into a product of a complete square
and a square-free part, and then split the square-free part into prime factors.
Assuming this representation, we claim the following:
\begin{thm}\label{solve:general:xi}
For any $\xi \in \Zt$ such that $\xi>0,\xi^{\bullet}>0$
there exists a factorization
\begin{equation}
\xi=\eta^{2}\xi_{1}...\xi_{r},\eta,\xi_{j}\in\Zt,r\in\Z,r\geq0,j=1,...,r\label{th3factorization}
\end{equation}
where $\xi_{1},...,\xi_{r}$ are $\Zt$-primes such that $\xi_j>0,\xi_j^{\bullet}>0, j=1,\ldots,r$.

Given such a factorization, the norm equation $N_{i}(x)=\xi$ is solvable in $\Ri$ if and
only if $p_{j}=N_{\t}(\xi_{j})$ is an integer prime for each $j\in\{1,...,r\}$ and either $p_{j} = 1 \m 5$ or $p_{j}= 5$.
\end{thm}

The proof is given in Appendix \ref{sec:app:B}.

In the context of Thm \ref{solve:general:xi}, we make the following:
\begin{observ}
$N_i(x)=\xi$ has at least $2^r$ distinct solutions.
\end{observ}
Assuming that each $N_i(y)=\xi_j$ is individually solvable, each factor equation has exactly two solutions.
The $\eta^2$ factor does not affect the solvability since $N_i(\eta)=\eta^2$, however, depending on the value of $\eta$, the $N_i(z)=\eta^2$ equation may have a number of other solutions besides $\eta$.
In general, solving factorization (\ref{th3factorization}) is as hard as factorizing an arbitrary rational integer.
This part of a general solution procedure cannot be done in polynomial running time, however we may consider solving norm equations where the righthand side happens to be factorizable at a lower cost.
An algorithm that solves a subclass of norm equations over $\Ri$ is summarized in Section \ref{normeq:algorithm:summary}.
We now present subalgorithms needed to implement Thm \ref{solve:prime:xi}.

\subsection{Square root modulo prime}

We present the following well-known fact (c.f., quadratic extensions in \cite{HCohen99}) as a segway into the modular square root algorithm:
\begin{thm}\label{solve:theModxi:theorem}
If $p=N_{\tau}(\xi)$ is an integer prime then $\Zt/(\xi)$ is effectively isomorphic
to $\mathbb{Z}_{p}$.
\end{thm}
\begin{proof}
Note that $\Zt/(\xi)$ is a field and that $p=\xi^{\bullet}\xi=0\m\xi$.
Therefore the natural embedding $\mathbb{Z}_{p}\rightarrow\Zt/(\xi),k\m p\rightarrow k\m\xi$
is well-defined.

Writing $\xi=a+b\tau,a,b\in\Z$, we prove that $b\neq0\m p$. Indeed,
$p=a^{2}-a\, b-b^{2}$. If $p$ divides $b$, then $p$ divides $a^{2}$
hence $p$ divides $a$. Hence $p$ would be proportional to $p^{2}$ which
is impossible.

Thus $b$ is invertible $\m p$, i.e., $\exists b_{1}\in\Z:bb_{1}=1\m p=1\m\xi$.
Then $\tau=-a\, b_{1}\m\xi$ is congruent to an integer $\m\xi$,
hence any element of $\Zt$ is congruent to an integer $\m\xi$. Therefore
the above embedding is epimorphic and in fact an isomorphism.
\end{proof}

In the context of Thm \ref{solve:theModxi:theorem} we conclude that $\tau-2$ is congruent to $(-a\, b_{1}-2) \m\xi$, and that finding an integer $m$ such that $m^2 = \tau-2 \m\xi$ is equivalent to finding a $m$ such that $m^2 = (-a\, b_{1}-2) \m p$.
In view of Thm \ref{solve:prime:xi}, the existence of such $m$ is guaranteed whenever $p=5\,l+1, l \in \Z$ or $p=5$.
In this case, computing a square root of $-a\, b_{1}-2$ modulo $p$ is performed, constructively, using the \textit{Tonelli-Shanks Algorithm} (\cite{DShanks},\cite{HCohen96}, Sec.~1.5), also given in Fig.~\ref{fig:Tonelli:Shanks}.
The algorithm is known to be on average probabilistically linear in bit sizes of the radicand and $p$, and probabilistically quadratic in the bit size of $p$ in the worst case (c.f., \cite{EBach}).
For convenience, we present Thm \ref{solve:theModxi:theorem} and the Tonelli-Shanks algorithm in a single procedure called SPLITTING-ROOT, given in Fig.~\ref{fig:Splittig:Root}.

\begin{figure}[t]
\begin{algorithmic}[1]
\Require $n, p \in \Z$, assume $p$ is an odd prime and $n$ is a quadratic residue $\m p$.
\Procedure{TONELLI-SHANKS}{$n$,$p$}
\State{Represent $p-1$ as $p-1=q\,2^s, \, q$ odd.}
\If {s=1}
return $\pm n^{(p+1)/4} \m p$;
\EndIf
\State By randomized trial select a quadratic non-residue $z$ , i.e. $z \in \{2,\ldots,p-1\}$ such that $z^{(p-1)/2} = -1 \m p$.
\State Let $c=z^q \m p$.
\State Let $r=n^{(q+1)/2} \m p, \, t=n^q \m p, \, m = s$.
\While {$t \neq 1 \m p$}
\State By repeated squaring find the smallest $i \in \{1,\ldots,m-1\}$ such that $t^{2^i}=1 \m p$
\State Let $b = c^{2^{m-i-1}} \m p$, $r \leftarrow r \, b, \, t \leftarrow t \, b^2, \, c\leftarrow b^2, \, m \leftarrow i$
\EndWhile
\EndProcedure
\Ensure $\{r, p-r\}$.

\end{algorithmic}
\caption{\label{fig:Tonelli:Shanks} Tonelli-Shanks algorithm.}
\end{figure}

\begin{figure}[t]
\begin{algorithmic}[1]
\Require $\xi \in \Zt$, assume $p=N_\t(\xi)$ is an odd prime and $p=1 \m 5$.
\Procedure{SPLITTING-ROOT}{$\xi=a + b \, \t$}
\State{$p \gets N_\t(\xi)$}, assert $b \neq 0 \m p$
\State{$b_1 \gets b^{-1} \m p$}
\State \Return TONELLI-SHANKS($-a\,b_1-2$,$p$)
\EndProcedure
\caption{\label{fig:Splittig:Root} Procedure SPLITTING-ROOT: Square root of $\t-2$ modulo $\xi$.}
\end{algorithmic}
\end{figure}

\subsection{Greatest common divisor in $\Ri$}  \label{subsection:GCD:Zomega}

Next we need an algorithm for computing the greatest common divisor in the $\Ri$ ring.
We consider a ``generalized binary" greatest common divisor algorithm for $\Ri$ that draws on ideas from \cite{DamFran}
and implements the general method given in \cite{Wikstrom}.
It requires the following:
\begin{lem}
For any  element of $\eta \in \Ri$, either $1+\w$ divides $\eta$ or $\eta$ is associate to an element $\zeta \in \Ri$ such that $\zeta=1 \m (1+\omega)$.
\end{lem}

\begin{proof}
Any element is congruent to a rational integer $\m (1+\w)$.
Since $5=N(1+\w)=(1+\w^3)(2-\w+\w^2-\w^3)(1-\w^2)(1+\w)$, it follows that $\eta \m (1+\w) = \eta \m (1+\w) \m 5$.
Thus any element is congruent to one of  $\{0,\pm 1, \pm 2\}$ modulo $1+\w$.
Since $-1$ is a unit, any element is associate to one that is congruent to either $0$, $1$ or $2$.
It remains to note that $2=u\,1+(1+\w)$, where $u=1-\w$ is a unit and therefore $2$ is associate to $1 \m (1+\w)$.
\end{proof}
Note that the integer remainder of an element $\eta \in \Ri$ modulo $(1+\w)$ is computed by substituting $-1$ for $\w$ in $\eta$ .
Computing the quotient and remainder of $\eta$ with respect to $(1+\w)$ requires a total of no more than 8 integer additions.
After reducing the remainder $\m 5$, the selection of a unit needed to associate $\eta$ with the desired $\zeta$ is straightforward and immediate.

The ``generalized binary" GCD algorithm based on the above Lemma is presented in Fig.~\ref{fig:GCD:Zomega}.
It follows from the analysis in \cite{Wikstrom} that this algorithm converges in a number of steps (or recursion depth) that is quadratic in bit sizes of the norms of inputs (and hence the number of steps is polylogarithmic in the magnitudes of the norms).

\begin{figure}[t]
\begin{algorithmic}[1]
\Require $a,b \in \Ri$
\Procedure{BINARY-GCD}{$a$,$b$}
\If {one of the inputs is zero}
\State{return the other input}
\EndIf
\If {$(1+\w)$ divides both inputs}
\State {compute $a_1$: $a=(1+\w)a_1$;
compute $b_1$: $b=(1+\w)b_1$;
return $(1+\w)$BINARY-GCD$(a_1,b_1)$}
\EndIf
\State $u=v=1$
\If {$(1+\w)$ divides neither of the inputs}
\State {select unit $u$ such that $u \,a = 1\m (1+\w)$;
select unit $v$ such that $v \,b = 1\m (1+\w)$}
\EndIf
\State {let $c \in \{a,b\}$ be the input with smaller norm}
\Ensure BINARY-GCD$(c,u\,a-v\,b)$.
\EndProcedure
\end{algorithmic}
\caption{\label{fig:GCD:Zomega} Procedure BINARY-GCD: GCD in $\Ri$.}
\end{figure}

\subsection{Discrete logarithm base $\t$}

Finally, we need an algorithm for the discrete logarithm of a unit in $\Zt$.
For completeness, we prove the following elementary lemma and then derive the desired algorithm from the proof.
\begin{lem}\label{UNITS:LEMMA}
(1) The group of units $U(\Zt)$ is generated by $-1$ and $\t$.

(2) If a unit $u\in U(\Zt)$ is such that $u>0,u^{\bullet}>0$ then
$u$ is a perfect square in $U(\Zt)$.
\end{lem}

\begin{proof}

(1) Both $-1,\t$ are clearly units.
Let $u=a+b\,\t\in U(\Zt)$ and define $\mu(u)=a\, b$. We show that there exists a certain $k\in\Z$ and $\delta=\pm1$ such
that $|\mu(u\,\delta\,\t^{k})|\leq1$.

Since $-1$ is a unit we can assume w.l.o.g. that $a>0$.
Now, consider the case of $\mu(u)>1$, implying $b>0.$ Since $N_{\t}(u)=(a-b)(a+b)-a\, b=\pm1$
and $a-b=(a\, b\pm1)/(a+b)$ it follows that $a>b$. Consider the
new unit $u'=u\,\t=a'+b'\t$ where $a'=b,b'=a-b$. We observe that
$a'>0,b'>0$ and $0<\mu(u')=a\, b-b^{2}<\mu(u)$. Thus $\mu(u)$
strictly decreases when the unit is multiplied by $\t$ but will not
become $<1$, as long as $\mu(u)>1$ . Therefore, there exists a positive
integer $k$ such that $\mu(u\,\t^{k})=1$ .

The case of $\mu(u)<-1$ is handled similarly, however, we repeatedly
multiply the unit times $\t^{-1}$instead of $\t$.
Units with $|a\, b|\leq1$ are easily enumerated and are found to be
$\{\pm1,\pm\t,\pm(1-\t)=\pm\t^{2},\pm(1+\t)=\pm\t^{-1}\}$.

(2) For a unit to $u$ to be positive, $u$ must be of the form $u=\t^{m},m\in\Z$.
Then $u^{\bullet}=(-(\t+1))^{m}$ is positive if and only if $m$
is even. Thus $u=(\t^{m/2})^{2}$.
\end{proof}

This proof suggests a straightforward algorithm for ``decoding" a unit, called UNIT-DLOG,
presented in Fig.~\ref{fig:Unit:Representation}.

\begin{figure}[h]
\begin{algorithmic}[1]
\Require unit $u = a + b\, \t \in U(\Zt)$
\Procedure{UNIT-DLOG}{$u$}
\State $s\gets 1$, $k \gets 0$
\If {$a < 0$}
\State{$a \leftarrow -a$; $b \leftarrow -b$; $s \leftarrow -s$ }
\EndIf
\State $\mu\gets a\,b$
\While {$|\mu|>1$}
\If {$\mu > 1$}
\State { $(a,b)\leftarrow (b,a-b)$; $k \leftarrow k-1$ }
\Else
\State { $(a,b)\leftarrow (a,a-b)$;  $k \leftarrow k+1$ }
\EndIf
\State $\mu \leftarrow a\,b$
\EndWhile
\Comment {$|\mu|=1$ here}
\State match $v=a+b \, \t$ with one of the $\{\pm1,\pm\t,\pm\t^{2},\pm\t^{-1}\}$
\Statex[1] adjust $s,k$ accordingly
\State \Return $(s,k)$
\EndProcedure
\Ensure $(s ,k)$ such that $s=-1,1,k$~-- integer and $u=s \t^k$
\end{algorithmic}
\caption{\label{fig:Unit:Representation} Procedure UNIT-DLOG. Finds a discrete logarithm of the unit $u$. The procedure runtime is in $O(\log(\max\{|a|,|b|\}))$.}
\end{figure}

\subsection{The EASY-SOLVABLE predicate}
The combination of Thms \ref{solve:prime:xi} and \ref{solve:general:xi} yields a principled constructive description of solutions of a norm equation over $\Ri$ where the only computationally hard part is the factorization of the righthand side of the equation.
%
%
A definition of what is easy to solve depends on how good we are at factorization in $\Zt$.
We make this dependency explicit in the algorithm presented in Fig.~\ref{fig:solve-norm-equation} and give an example of a viable EASY-FACTOR procedure, also given in Fig.~\ref{fig:easy-factor}; as future work, further enhancements of EASY-FACTOR could lead to even better compiled circuits.

Procedure EASY-SOLVABLE (Fig.~\ref{fig:easy-solvable}) has a single input that is assumed to be a list of factors belonging to $\Zt$ with their multiplicities.
A factor of the form $\eta^{2s}, s\in \Z$, contributes a factor of $\eta^s$ to the overall solution and its presence or absence does not affect the solvability of the equation.
A factor of multiplicity $1$ is either hard to factorize or it is prime.
As per Thm \ref{solve:prime:xi}, a prime factor $\xi \in \Zt$ is potentially a witness that the overall equation is not solvable, unless $p=N_{\t}(\xi)$ is an integer prime and either $p=5$ or $p=1 \m 5$.
We use a primality test~(subroutine IS-PRIME in procedure EASY-SOLVABLE) that has probabilistically polynomial runtime and negligible probability of returning a false positive.
Procedure EASY-FACTOR in Fig.~\ref{fig:easy-factor} is an example of a minimum-effort factorizer that is sufficient for our purposes.

\begin{figure}[h]
\begin{algorithmic}[1]
\Require $fl : List\langle \Zt \times \Z \rangle$
\Procedure{EASY-SOLVABLE}{$fl$}
\For {$i \in \{0..length(fl)-1\}$}
\State match $fl[i]$ with $(\xi,k), \xi \in \Zt, k \in \Z$
\If {$k = 1 \m 2$}
\If {$\xi \neq 5$}
\State $p\gets N_{\t}(\xi)$
\State $r\gets p \m 5$
\If {not  IS-PRIME($p$) or $r \notin \{0,1\}$)}
\State \Return FALSE
\EndIf
\EndIf
\EndIf
\EndFor
\State \Return TRUE
\EndProcedure
\Ensure TRUE if $fl$ is a factorization of an easy and solvable instance of the equation.
\end{algorithmic}
\caption{\label{fig:easy-solvable} Procedure EASY-SOLVABLE: checks if the given instance of the norm equation can be solved in polynomial time.}
\end{figure}

\begin{figure}[h]
\begin{algorithmic}[1]
\Require $\xi \in \Zt$
\Procedure{EASY-FACTOR}{$\xi=a + b\, \t , a,b \in \Z$}
\State $c \gets GCD(a,b); a_1 \gets a/c; b_1 \gets b/c; \xi_1 \gets a_1 + b_1 \, \t$
\If {$c=d^2 , d \in \Z$}
\State $ret \gets List((d,2))$
\Else
\If { $c=5\,d^2, d \in \Z$ }
\State $ret \gets List((d,2),(5,1))$
\Else
\State \Return $List((\xi,1))$
\Statex[1] \Comment equation is not going to be solvable
\EndIf
\EndIf
\State $n \gets N_\t(\xi_1)$
\If {$n=0 \m 5$}
\State $\xi_2 \gets \xi_1/(2-\t)$
\State \Return $ret+((2-\t),1)+(\xi_2,1)$
\Else
\State \Return $ret+(\xi_1,1)$
\EndIf
\EndProcedure
\Ensure Returns lightweight factorization of input.
\end{algorithmic}
\caption{\label{fig:easy-factor} Procedure EASY-FACTOR: a minimum-effort factorizer.}
\end{figure}

\subsection{The algorithm} \label{normeq:algorithm:summary}

Using the described procedures, we present an algorithm for solving norm equations, SOLVE-NORM-EQUATION, given in Fig.~\ref{fig:solve-norm-equation}.
The algorithm first checks the necessary conditions $\xi > 0 ,  \xi^{\bullet} > 0$ on the righthand side of the equation, and provided the conditions are satisfied, invokes EASY-FACTOR to preprocess $\xi$.
If the resulting list of factors is EASY-SOLVABLE, we consider each factor to either have even multiplicity or be a power of an allowed prime in $\Zt$.
An allowed prime is either $5$ or $2-\t$ with norm $5$, or some other prime with norm $p$ such that $p=1 \m 5$.
$5$ is equal to the norm of $2\,\t+1$.
In the case of $2-\t$, we exploit the identity $|\w+\w^4|^2 = 2-\t$ and induce the factor $(\w+\w^4)$ into the solution.
In the more general case we need to perform all the steps prescribed by Thm \ref{solve:prime:xi}, specifically:
(1) represent $\t-2$ as a square of integer $M$ modulo $\xi$  (as in Fig.~\ref{fig:Splittig:Root}),
(2) compute the GCD $y$ of $\xi$ and $M-(\w+\w^4)$ in $\Ri$ (as in Fig.~\ref{fig:GCD:Zomega}) ,
(3) obtain a unit $u$ that associates $\xi$ with $|y|^2$,
(4) represent the unit $u$ as $\t^m$ using the UNIT-DLOG procedure (as in Fig.~\ref{fig:Unit:Representation}).
Having performed these steps, we a obtain a factor of the desired solution corresponding to the allowed prime factor of the righthand side.
The algorithm terminates when all the factors of the righthand side have been inspected.

\begin{figure}[h]
\begin{algorithmic}[1]
\Require $\xi \in \Zt$
\Procedure{SOLVE-NORM-EQUATION}{$\xi$}
\If {$\xi < 0 $ or $\xi^{\bullet} < 0$}
\State \Return UNSOLVED
\EndIf
\State $fl\gets$EASY-FACTOR($\xi$)
\If {not EASY-SOLVABLE($fl$)}
\State \Return UNSOLVED
\EndIf
\State $x \gets 1$
\For {$i \in \{0..length(fl)-1\}$}
\State match $fl[i]$ with $(\xi_i,m) , \xi_i \in \Zt, m \in \Z$
\State $x \gets x \, \xi_i^{m/2}$
\If {$m=1 \m 2$}
\Comment assert $\xi_i$ is easy factor
\If {$\xi_i = 5$}
\State $x \gets x \, (2\,\t+1)$
\Else
\If {$\xi_i = 2-\t$}
\State $x \gets x \, (\w+\w^4)$
\Else
\State $M\gets\text{SPLITTING-ROOT}(\xi_i)$
\Statex[1] \Comment $M^2 = \t-2 \m \xi_i$
\State $y\gets\text{BINARY-GCD}(\xi_i,M{-}(\w{+}\w^4))$
\Statex[1] \Comment $(\w+\w^4)^2=\t-2$
\State $u\gets \xi_i/|y|^2$
\Statex[1] \Comment $u$ -- unit, $u>0$, $\gc{u}>0$
\State $(s,m)\gets\text{UNIT-DLOG}(u)$
\Statex[1] \Comment $s=1,$ $m$ -- even
\State $x \gets x \, \t^{m/2} y$
\EndIf
\EndIf
\EndIf
\EndFor
\State \Return $x$
\EndProcedure
\Ensure $x$ from $\Zw$ such that $|x|^2=\xi$
\end{algorithmic}
\caption{\label{fig:solve-norm-equation}Procedure SOLVE-NORM-EQUATION: finds a solution to an ``easy'' instance of the norm equation in probabilistic polynomial time.}
\end{figure}

\begin{example}
Consider the relative norm equation
\begin{equation}
N_i(x)=\xi=760-780\,\t.
\end{equation}
This equation turns out to be EASY-SOLVABLE with one of the solutions
\begin{equation} \label{example:normeq:soluiton}
x=2\,(4+3\,\t)(12-20\,\w+15\,\w^2-3\,\w^3).
\end{equation}
\end{example}

Following the EASY-FACTOR procedure, it is relatively easy to obtain the following list of factors for $\xi$:
\[fl=\{(2,2),(5,1),((2-\t),1),((15-8\,\t),1)\}.
\]
All but the last factor in this list yield easy partial solutions: $N_i(2)=2^2$, $N_i(2\,\t+1)=5$,$N_i(\w+\w^4)=2-\t$.
For the last factor, we find that $p=N_\t(15-8\,\t)=281$ is prime and $ p = 1 \m 5$.
SPLITTING-ROOT($15-8\,\t$) yields $63$ and BINARY-GCD($15-8\,\t$, $63-(\w+\w^4)$) yields $y=3+2\,\w-7\,\w^2+7\,\w^3$.
By direct computation, $(15-8\,\t)/|y|^2 = 5+3\,\t = \t^{-4}$ and thus
$N_i(\t^{-2}\, (3+2\,\w-7\,\w^2+7\,\w^3))=15-8\,\t$.
The value in (\ref{example:normeq:soluiton}) is a routine simplification of $2\,(2\,\t+1)\,\t^{-2}(\w+\w^4)(3+2\,\w-7\,\w^2+7\,\w^3)$.
(Note that further simplification is possible using $\t=\w-\w^3$ and is left as an exercise.)

\section{Approximation }
\label{sec:appr}

In this section we combine methods developed in previous sections and describe an algorithm for approximating unitaries of the form $R_z(\phi)X$ and $R_z(\phi)X$ with $\FT$-circuits (or equivalently, $\sgm$-circuits).
Recall that we measure the quality of approximation $\ve$ using the global phase-invariant distance
\begin{equation}
d(U,V)=\sqrt{1-\l|tr(UV^{\dagger})\r|/2}.
\end{equation}
The quality of approximation $\ve$ defines the problem size. Our algorithm produces circuits of length $O(\log(1/\ve))$ which meets the asymptotic worst-case lower bound~\cite{HRC} for such a circuit. The algorithm is probabilistic in nature and on average requires running  time in $O(\log^{c}(1/\ve))$ to find an approximation where $c$ is  a constant close to but smaller than $2.0$, according to our empirical estimates .

We first discuss all details of the algorithm for approximating $R_z(\phi)$ and then show how the same tools allow us to find approximations of $R_z(\phi)X$.

There are two main stages in our algorithm: the first stage approximates $R_z(\phi)$ with an exact unitary $U[u,v,0]$ and then uses the exact synthesis algorithm~(Figure~\ref{fig:exact:synthesis}) to find a circuit implementing $U[u,v,0]$. The second stage is completely described in Section~\ref{sec:exact}; here we focus on the first stage.

The expression for the quality of approximation in the first case simplifies to
\[
 d(R_z(\phi),U[u,v,0])=\sqrt{1-\l|\re(ue^{i\phi/2})\r|}
\]
We see that the quality of approximation depends only on $u$, the top left entry of $U[u,v,0]$. Therefore, to solve the first part of the problem it is sufficient to find $u$ from $\Zw$ such that $\sqrt{1-\l|\re(ue^{i\phi/2})\r|}\le\ve$. However, there is an additional constraint that must be satisfied: there must exist a $v$ from $\Zw$ such that $U[u,v,0]$ is unitary, in other words the following equation must be solvable:
\begin{equation}
  |v|^2 = \xi,\text{ for } \xi=\p(1-|u|^2).
\end{equation}
This is precisely the equation studied in Section~\ref{sec:norm-equation}. As discussed, in general the problem of deciding whether such a $v$ exists and then finding it is hard. It turns out, however, that in our case there is enough freedom to pick (find) ``easy" instances and obtain a solution in polynomial time without sacrificing too much quality.

There is an analogy to this situation: it is well known that factoring a natural number into prime factors is a hard problem. However, checking that the number is prime can be done in polynomial time. In other words, given a natural number $N$ one can efficiently decide if it is an easy instance for factoring. Now imagine the following game: one is given a uniformly chosen random number from the interval $[0,N]$ and one wins each time they can factor it. How good is this game? The key here is the Prime Number Theorem. It states that there are $\Theta(N/\log(N))$ primes in the interval $[0,N]$. Therefore one can win the game with probability at least $\Omega(1/\log(N))$. In other words, the number of trials one needs to make before winning scales as $O(\log(N))$. In our case the situation is somewhat similar and $N$ is of order $1/\ve$.

At a high level, during our approximation procedure~(Figure~\ref{fig:rzappr:algorithm}) we perform a number of trials. During each trial we first randomly pick a $u$ from $\Zw$ that achieves precision $\ve$ and then check that the instance of the norm equation can be easily solved. Once we find such an instance we compute $v$ and construct a unitary~$U[u,v,0]$.

\begin{figure}[t]
\begin{algorithmic}[1]
\Require $\phi$ -- defines $R_z(\phi)$, $\ve$ -- precision
\State $C \gets \sqrt{\p/4}$
\State $m \gets \l\lceil \log_{\tau}(C\ve)\r\rceil+1$
\State Find $k$ such that $\th = -\phi/2-\pi k/5 \in [0,\pi/5]$
\State \emph{not-found} $\gets$ true, $u\gets0,v\gets0$
\While{\emph{not-found}}
  \State $u_0 \gets \text{RANDOM-SAMPLE}(\theta,\ve,1)$ \Comment See Figure~\ref{fig:appr:sampling}
  \State $\xi \gets \p\l(\p^{2m} - \l|u_0\r|^2\r)$
  \State $fl \gets \text{EASY-FACTOR}(\xi)$
  \If{$\text{EASY-SOLVABLE}(fl)$}
    \State \emph{not-found} $\gets$ false
    \State $u \gets \w^k\t^m u_0$
    \State $v \gets \t^m\text{SOLVE-NORM-EQUATION}(\xi)$
  \EndIf
\EndWhile
\State $C \gets \text{EXACT-SYNTHESIZE}(U[u,v,0])$
\Ensure Circuit $C$ such that $d(C,R_z(\phi))\le\ve$
\end{algorithmic}
\caption{\label{fig:rzappr:algorithm} The algorithm for approximating $R_z(\phi)$ by an $\FT$-circuit with $O(\log(1/\ve))$ gates and precision at most $\ve$. Runtime is probabilistically polynomial as a function of $\log(1/\ve)$. }
\end{figure}

We generate a random element $u$ from $\Zw$ that has the desired precision using procedure RANDOM-SAMPLE. To achieve a better constant factor in front of $\log(1/\ve)$ for the length of the circuit we randomly chose $u_0=u\p^m$ instead of $u$. It is easy to recover $u$ as $\t=\p^{-1}$ and $u=u_0\t^n$.

In Figure \ref{fig:appr:eps-region}, when $r=1$, the light gray circular segment corresponds to such $u_0$ that $U[u,v,0]$ is within $\ve$ from $R_z(\phi)$. The element $u_0$ is a complex number and, as usual, the $x$-axis of the plot corresponds to the real part and the $y$-axis to the imaginary part. All random samples that we generate belong to the dark gray parallelogram and have the form $a_x+b_x\t+i\sqrt{2-\t}(a_y+b_y\t)$~(note that $i\sqrt{2-\t}$ is equal to $\w+\w^4$ and belongs to $\Zw$). We first randomly choose an imaginary part and then a real part. To find an imaginary part
we randomly choose a real number $y$ and then approximate it with $\sqrt{2-\t}(a_y+b_y\t)$ using the APPROX-REAL~(Figure \ref{fig:appr:algorithm})
procedure. Once we
find $\sqrt{2-\t}(a_y+b_y\t)$, we choose the $x$-coordinate as shown in Figure \ref{fig:appr:eps-region} and approximate it with $a_x+b_x\t$.

\begin{figure}[t]
\begin{algorithmic}[1]
\Require $\theta$ -- angle between $0$ and $\pi/5$, $\ve$ -- precision, $r\ge1$
\Procedure{RANDOM-SAMPLE}{$\theta,\ve,r$}  \Comment See Fig. \ref{fig:appr:eps-region}
\State $C \gets \sqrt{\p/(4r)}$
\State $m \gets \l\lceil \log_{\tau}(C\ve r)\r\rceil+1$
\State $N \gets \l\lceil \p^m \r\rceil$
\State $\ym \gets r\p^m(\sin(\th )-\ve\l(\sqrt{4-\ve ^2}\cos(\th )+\ve\sin(\th)\r)/2)$
\State $\yM \gets r\p^m(\sin(\th )+\ve\l(\sqrt{4-\ve ^2}\cos(\th )-\ve\sin(\th)\r)/2)$
\State $\xM \gets r\p^m((1 - \ve^2/2)\cos(\th) - \ve\sqrt{1-\ve^2/4}\sin(\th))$
\State $x_c \gets \xM - r\ve^2\p^m/(4\cos(\theta)) $
\State Pick random integer $j$ from $[1,N-1]$
\State $y \gets \ym+j(\yM-\ym)/N$
\State $a_y+\t b_y \gets \text{APPROX-REAL}(y/\sqrt{2-\t},m)$ \Comment Fig.~\ref{fig:appr:algorithm}
\State $x \gets x_c - ((a_y+b_y\t)\sqrt{2-\t} - \ym)\tan(\th)$
\State $a_x+\t b_x \gets \text{APPROX-REAL}(x,m)$ \Comment Fig.~\ref{fig:appr:algorithm}
\State \Return $a_x+\t b_x + \sqrt{\t-2}(a_y+\t b_y)$
\EndProcedure
\end{algorithmic}
\caption{\label{fig:appr:sampling} The algorithm for picking a random element of $\Zw$ that is in the dark gray region in Figure \ref{fig:appr:eps-region}. Number of different outputs of the algorithm is in $O(1/\ve)$.
}
\end{figure}

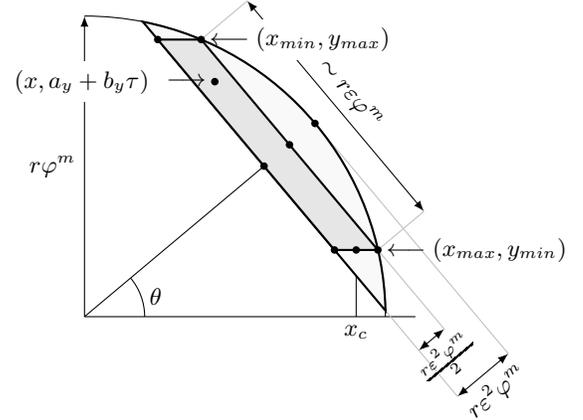
\begin{figure}[t]
\newcommand*{\fTh}{40}%
\newcommand*{\fEps}{0.47}%

\begin{tikzpicture}

\pgfmathsetmacro{\fsc}{4}
\pgfmathsetmacro{\fxmin}{cos(\fTh)-1/2*\fEps*\fEps*cos(\fTh)-1/2*\fEps*sqrt(4-\fEps*\fEps)*sin(\fTh)}%
\pgfmathsetmacro{\fxmax}{cos(\fTh)-1/2*\fEps*\fEps*cos(\fTh)+1/2*\fEps*sqrt(4-\fEps*\fEps)*sin(\fTh)}%
\pgfmathsetmacro{\fymin}{sin(\fTh)-1/2*\fEps*\fEps*sin(\fTh)-1/2*\fEps*sqrt(4-\fEps*\fEps)*cos(\fTh)}%
\pgfmathsetmacro{\fymax}{sin(\fTh)-1/2*\fEps*\fEps*sin(\fTh)+1/2*\fEps*sqrt(4-\fEps*\fEps)*cos(\fTh)}%

\pgfmathsetmacro{\fxminf}{cos(\fTh)-1/2*\fEps*\fEps*2*cos(\fTh)-1/2*\fEps*sqrt(2)*sqrt(4-\fEps*\fEps*2)*sin(\fTh)}%
\pgfmathsetmacro{\fxmaxf}{cos(\fTh)-1/2*\fEps*\fEps*2*cos(\fTh)+1/2*\fEps*sqrt(2)*sqrt(4-\fEps*\fEps*2)*sin(\fTh)}%
\pgfmathsetmacro{\fyminf}{sin(\fTh)-1/2*\fEps*\fEps*2*sin(\fTh)-1/2*\fEps*sqrt(2)*sqrt(4-\fEps*\fEps*2)*cos(\fTh)}%
\pgfmathsetmacro{\fymaxf}{sin(\fTh)-1/2*\fEps*\fEps*2*sin(\fTh)+1/2*\fEps*sqrt(2)*sqrt(4-\fEps*\fEps*2)*cos(\fTh)}%
\pgfmathsetmacro{\fxw}{\fEps*\fEps/2/cos(\fTh)*\fsc}%
\pgfmathsetmacro{\fxp}{cos(\fTh)*\fsc}%
\pgfmathsetmacro{\fyp}{sin(\fTh)*\fsc}%
\pgfmathsetmacro{\fxps}{cos(\fTh)*\fsc*(1-\fEps*\fEps)}%
\pgfmathsetmacro{\fyps}{sin(\fTh)*\fsc*(1-\fEps*\fEps)}%
\pgfmathsetmacro{\fxpsh}{cos(\fTh)*\fsc*(1-\fEps*\fEps/2)}%
\pgfmathsetmacro{\fypsh}{sin(\fTh)*\fsc*(1-\fEps*\fEps/2)}%
\pgfmathsetmacro{\fxth}{cos(\fTh/2)*\fsc*0.2}%
\pgfmathsetmacro{\fyth}{sin(\fTh/2)*\fsc*0.2}%
\pgfmathsetmacro{\fxc}{\fxmax*\fsc-0.5*\fxw}%
\pgfmathsetmacro{\fy}{(\fymin * 0.2 + \fymax* 0.8)*\fsc}%
\pgfmathsetmacro{\fx}{(\fxmax *0.2 + \fxmin*0.8)*\fsc-0.5*\fxw }%
\pgfmathsetmacro{\fyoff}{\fsc*cos(\fTh)}%
\pgfmathsetmacro{\fxoff}{\fsc*sin(\fTh)}%

\coordinate (center) at (0,0);
\coordinate (1) at (\fxp,\fyp);
\coordinate (2) at (\fxps,\fyps);
\coordinate (1a) at (\fxp+\fxoff,\fyp-\fyoff);
\coordinate (2a) at (\fxps+\fxoff,\fyps-\fyoff);
\coordinate (2b) at (\fxps+\fxoff*0.8,\fyps-\fyoff*0.8);
\coordinate (3) at (\fxminf*\fsc,\fymaxf*\fsc);
\coordinate (4) at (\fxmaxf*\fsc,\fyminf*\fsc);
\coordinate (5) at (\fsc,0);
\coordinate (6) at (\fxmin*\fsc,\fymax*\fsc);
\coordinate (7) at (\fxmax*\fsc,\fymin*\fsc);
\coordinate (6a) at (\fxmin*\fsc+\fyoff*0.2,\fymax*\fsc+\fxoff*0.2);
\coordinate (7a) at (\fxmax*\fsc+\fyoff*0.2,\fymin*\fsc+\fxoff*0.2);
\coordinate (8) at (\fxmax*\fsc-\fxw,\fymin*\fsc);
\coordinate (9) at (\fxmin*\fsc-\fxw,\fymax*\fsc);
\coordinate (xc) at (\fxc,\fymin*\fsc);
\coordinate (s) at (\fx,\fy);
\coordinate (10) at (\fxpsh,\fypsh);
\coordinate (10a) at (\fxpsh+\fxoff*0.8,\fypsh-\fyoff*0.8);

\tikzset{>=latex}
\draw (5) arc(0:90:\fsc);
\draw[->] (center) -- (0,\fsc);
\draw (center)--(1);
\draw (center)--(\fsc*1.1,0);
\draw (xc)--(\fxc,0);
\draw[draw=gray!50] (xc)--(s);
\draw[draw=gray!50] (1)--(1a);
\draw[draw=gray!50] (2)--(2a);
\draw[draw=gray!50] (10)--(10a);
\draw[<->] (2a)--(1a) node [midway,sloped,below] {$r\ve^2\p^m$};
\draw[<->] (10a)--(2b) node [midway,sloped,below] {$\frac{r\ve^2\p^m}{2}$};

\draw[<->] (6a)--(7a)  node [midway,sloped,above] {$\sim r\varepsilon\p^m$};;
\node[right,fill=white] at (6) {$\longleftarrow(x_{min},y_{max})$};
\draw[draw=gray!50] (6)--(6a);
\draw[draw=gray!50] (7)--(7a);

\draw(0.2*\fxp,0.2*\fyp) arc(\fTh:0:\fsc*0.2);
\filldraw[draw=black, fill=gray!5 ,thick]
   let   \p1 = ($(3) - (center)$),
       \p2 = ($(4) - (center)$),
       \n0 = {veclen(\x1,\y1)},
       \n1 = {atan(\y1/\x1)},
       \n2 = {atan(\y2/\x2)}
    in
     (3) arc(\n1:\n2:\n0)  -- cycle;

\filldraw[draw=black, fill=gray!20, thick]
     (6) -- (7) -- (8) -- (9)  -- cycle;

\node[right] at (7) {$\longleftarrow(x_{max},y_{min})$};
\node[right] at (\fxth,\fyth) {$\theta$};
\node[below] at (\fxc,0) {$x_c$};
\node[left] at (0,\fsc/2) {$r\p^m$};
\node[fill=white] at (-0.2,\fy) {$\quad(x,a_y+b_y\t)$};
\node[left] at (s) {$\longrightarrow$};

\foreach \dot in {1,2,6,7,8,9,xc,s,10} {
  \fill (\dot) circle(0.05);

}
\end{tikzpicture}

\caption{\label{fig:appr:eps-region}An $\ve$-region and visualization of variables used in RANDOM-SAMPLE procedure (Figure~\ref{fig:appr:sampling}). }
\end{figure}

The problem of approximating real numbers with numbers of the form $a+bz$ for integers $a,b$ and irrational $z$ is well studied. The main tool is continued fractions.
It is also well known that Fibonacci numbers $\{F_n\}$,
\[
F_0=0, F_1=1, F_n=F_{n-1}+F_{n-2}, n \ge 2,
\]
 are closely related to the continued fraction of the Golden section $\p$ and its inverse $\t$. The correctness of procedure APPROX-REAL~(Figure \ref{fig:appr:algorithm}) is based on the following very well-known result connecting $\t$ and the Fibonacci numbers; we state it in a convenient form and provide the proof for completeness.
\begin{prop} \label{appr:prop:cfr}
For any integer $n$
\[
 \l|\t-\frac{F_n}{F_{n+1}}\r| \le \frac{\t^n}{F_{n+1}}
\]
and $F_n \ge (\p^n-1)/\sqrt{5}$.
\end{prop}
\begin{proof}
First we define the family of functions $\{f_n\}$ such that $f_1(x)=x$ and $f_n(x)=1/(1+f_{n-1}(x))$ for $n\ge2$. It is not difficult to check that $f_n(\t)=\t$. It can be shown by induction on $n$ that $f_n(1)$ is equal to $F_n/F_{n+1}$. Therefore, to prove the first part of the proposition we need to show that
\[
\l|f_n(\t) - f_n(1) \r| \le \t^n / F_{n+1}.
\]
We proceed by induction. The statement is true for $n$ equal to $1$. Using the definition of $f_n$ it is not difficult to show
\[
\l|f_{n+1}(\t) - f_{n+1}(1) \r| = f_{n+1}(\t)f_{n+1}(1)\l|f_n(\t) - f_n(1) \r|.
\]
Using $f_{n+1}=\t$ and the inequality for $\l|f_n(\t) - f_n(1) \r|$ we complete the proof of the first part.

To show the second part it is enough to use the well-known closed form expression for Fibonacci numbers $F_n = (\p^n - (-\t)^n)/\sqrt{5}$.

\end{proof}
\begin{figure}[t]
\begin{algorithmic}[1]
\Require $x$ -- real number, $n$ -- defines precision as $\t^{n-1}(1-\t^n)$
\Procedure{APPROX-REAL}{$x$,$n$}
\State $p \gets F_n, q \gets F_{n+1}$ \Comment $F_n$ -- Fibonacci numbers
\State $u \gets {(-1)}^{n+1} F_n, v \gets (-1)^{n} F_{n-1}$ \Comment $up+vq=1$
\State $c \gets \lfloor xq \rceil$
\State $a \gets cv + p \lfloor cu/q \rceil$
\State $b \gets cu - q \lfloor cu/q \rceil $
\State \Return $a+\t b$
\EndProcedure
\Ensure $a+b\t\text{ s.t. }|x-(a+b\t)| \le \t^{n-1}(1-\t^n), |b|\le\p^n$
\end{algorithmic}
\caption{\label{fig:appr:algorithm}The algorithm for finding integers $a,b$ such that $a+b\t$ approximates the real number $x$ with precision $\t^{n-1}(1-\t^n)$ }
\end{figure}
 At a high level, in procedure APPROX-REAL we approximate $\t$ with a rational number $p/q$ and find the approximation of $x$ by a rational of the form $a+bp/q$. To get good resulting precision we need to ensure that $b(\t-p/q)$ is small, therefore we pick $b$ in such a way that $|b|\le q/2$.  The details are as follows.
\begin{lem}\label{lem:appr:real}
Procedure APPROX-REAL (Figure~\ref{fig:appr:algorithm}) outputs integers $a,b$ such that $|x-(a+b\t)|\le\t^{n-1}(1-\t^n)$, $|b|\le\p^n$ and terminates in time polynomial in $n$.
\end{lem}
\begin{proof}
First, identity $F_{n+1}F_{n-1}-F_n^2= (-1)^n$ for Fibonacci numbers implies $uv+pq=1$. Next, by choice of $c$ we have $|xq-c|\le1/2$, therefore $|x-c/q|\le1/2q$. By Proposition~\ref{appr:prop:cfr}, $1/2q$ is less than $\t^n\sqrt{5}(1-\t^n)/2$; it remains to show that $c/q$ is within distance $\t^n/2$ from $a+b\t$ by the triangle inequality. By choice of $a,b$ we have $c=aq+bp$ and
\[
\l|c/q-(a+b\t)\r|= |b|\l|\t-p/q\r|.
\]
Using Proposition~\ref{appr:prop:cfr}, equality $q=F_{n+1}$ and inequality $|b|\le q/2$ we conclude that $\l|c/q-(a+b\t)\r|\le\t^n/2$.

The complexity of computing $n^{th}$ Fibonacci number is polynomial in $n$. Assuming that the number of bits used to represent $x$ is proportional to $n$ all arithmetic operations also have complexity polynomial in $n$.
\end{proof}
There are two main details of the RANDOM-SAMPLE procedure we need to clarify. The first is that the result is indeed inside the dark gray parallelogram on Figure~\ref{fig:appr:eps-region}. This is achieved by picking real values $x,y$ far enough from the border of the parallelogram and then choosing the precision parameter $m$ for APPROX-REAL in such a way that $a_x + b_x\t$(close to $x$) and $a_y+b_y\t$(close to $y/\sqrt{2-t}$) stay inside the parallelogram.

The second important detail is the size of resulting coefficients $a_x,b_x,a_y,b_y$. It is closely related to the number of gates in the resulting circuit. Therefore it is important to establish an upper bound on coefficients size. The following lemma provides a rigorous summary:
\begin{lem} \label{lem:appr:sampling}
When the third input $r\ge1$, procedure RANDOM-SAMPLE has the following properties:
\begin{itemize}
 \item there are  $O(1/\ve)$ different outputs and each of them occurs with the same probability,
 \item the procedure outputs an element $u_0$ of $\Zw$ from the dark gray parallelogram $P$ in Figure~\ref{fig:appr:eps-region},
 \item the Gauss complexity measure of $u_0$ is in $O(1/\ve)$.
\end{itemize}
\end{lem}
\begin{proof}
By construction, the algorithm produces $N-1$ outputs with equal probability. It is not difficult to check that $N$ is in $O(1/\ve)$. We first show that the outputs are all distinct and their $y$ coordinate is in $[\ym,\yM]$. This follows from an estimate
\[
 \l|y-(a_y+b_y\t)\r|\sqrt{2-\t}\le (\yM - \ym) / 2N
\]
because each randomly generated $y$ is at least distance $(\yM - \ym) / N$ from any other randomly generated $y$ and also $\ym,\yM$. To show the estimate we use the result of Lemma~\ref{lem:appr:real} and check that
\[
  \t^{m-1}(1-\t^m)\sqrt{2-\t} \le (\yM - \ym) / 2N
\]
which is straightforward, but tedious. If we concentrate only on terms that are first order in $\ve$ we get:
\begin{equation}\label{eq:appr:estimate}
  \t^{m-1}(1-\t^m) \lesssim C\ve r, (\yM - \ym) / 2N \gtrsim \ve\cos(\theta)r.
\end{equation}
The constraint on $\th$ gives $\cos(\theta) \ge \p/2$. From $C\sqrt{2-\t}<\p/2$ we conclude that the inequality is true for the terms that are first order in $\ve$.

To show that the procedure output belongs to the parallelogram $P$, it is sufficient to check that $a_x+b_x\t$ is within distance $\p^m r\ve^2/4$ (half of the parallelogram height) from $x_c - (a_y+b_y\t - \ym)\sin(\th)$. Again using Lemma~\ref{lem:appr:real}, it is sufficient to show that
\[
 \t^{m-1}(1-\t^m) \le \p^m r\ve^2/4.
\]
We again analyze the expression up to the first order terms in $\ve$. We note that $\p^m r \ve^2/4 \eqsim \ve/(4C\t)$; combining it with the inequality above, using (\ref{eq:appr:estimate}) and $C=\sqrt{\p/(4r)}$, we conclude that all outputs of the algorithm are inside parallelogram~$P$.

To show the last property, we note that by Lemma~\ref{lem:appr:real} both $|b_x|,|b_y|$ are bounded by $\p^m$. The same is true for $|a_x|,|a_y|.$ Indeed, $|x|,|y|$ are both of order $\p^m$ and
\[
 |a_y| \lesssim |y/\sqrt{2-\t}-b_x\t|+C\ve r,\,|x_y| \lesssim |x-b_x\t| + C\ve r.
\]
This implies that if we write $u_0$ as $\sum_{k}\alpha_k\w^k$ each integer $\alpha_k$ will be of order $\p^m$ which is the same as $O(1/\ve)$. Using the upper bound on the Gauss complexity measure $\G{u_0}$ in terms of $\alpha_k$ from Proposition \ref{prop:exact:alternatives} we conclude that $\G{u_0}$ is in $O(1/\ve)$.
\end{proof}

The technical tools that we have developed so far are sufficient to verify that our approximation algorithm achieves the required precision and produces circuits of length $O(1/\log(1/\ve))$. The remaining part is to show that on average the algorithm requires $O(\log^c(1/\ve))$ steps. It relies on the following conjecture, similar in nature to the Prime Number Theorem:
\begin{cnj} \label{cnj:appr:distr}
Let $\pi(M)$ be the number of elements $\xi$ from $\Zt$ such that $N_\t(\xi)$ is a prime representable as $5n+1$ and less than $M$, then $\pi(M)$ is in $\Theta(M/\log(M))$.
\end{cnj}

The conjecture defines the frequency of easy instances of the norm equation during the sampling process. Finally we prove the main theorem.

\begin{thm} \label{thm:appr:main}
Approximation algorithm~(Figure~\ref{fig:rzappr:algorithm}) outputs a $\FT$-circuit $C$ of length $O(\log(1/\ve))$ such that $d(C,R_z(\phi))\le\ve.$ On average the algorithm runtime is in $O(\log^c(1/\ve))$ if Conjecture \ref{cnj:appr:distr} is true.
\end{thm}
\begin{proof}
First we show that the algorithm achieves the desired precision and produces a $\FT$-circuit of length $O(\log(1/\ve))$. Both statements follow from Lemma~\ref{lem:appr:sampling}. It is not difficult to check that the light gray segment on Figure~\ref{fig:appr:eps-region} defines all $u_0$ such that $d(U[u_0\t^m\w^k,v,0],R_z(\phi))$ is less than $\ve$. Therefore, by picking samples from the dark gray parallelogram $P$, we ensure that we achieve precision~$\ve$. The value of the Gauss complexity measure is in $O(1/\ve)$, therefore by Theorem~\ref{thm:exact:correctness} the length of the resulting circuit is in $O(1/\log(\ve))$.

There are two necessary conditions for predicate EASY-SOLVABLE to be true:

(1) $N_\t(\xi)$ is prime and equals $5n+1$ for integer $n$,

(2) $\xi > 0, \gc{\xi} > 0.$

Let $p_M$ be the probability that the first condition is true when we choose $\xi$ uniformly at random and $N_\t(\xi)$ is bounded by $M$. Following Conjecture~\ref{cnj:appr:distr}, we assume that $p_M$ is in $O(1/\log(M))$. In our case $M$ is of order $\p^{2m}$ and therefore the probability of getting an instance solvable in polynomial time is in $O(1/\log(1/\ve))$.

Now we show that the second condition is satisfied by construction. Part $\xi>0$ is trivial because procedure RANDOM-SAMPLE always generates $u_0$ such that $|u_0|\le\p^m$. For the second part we use Proposition~\ref{prop:exact:alternatives} and note that for non-zero $u_0\t^{n}$ the value of the Gauss complexity measure is
\[
 \G{u_0\t^{n}} = \cm{(u_0\t^{n})}+\Ni{u_0\t^{n}} \ge 2.
\]
We conclude that  $\cm{(u_0\t^{n})} \ge 1$ which gives
\[
\gc{\xi}=\t^{2m+1}(\cm{(u_0\t^{n})}-1) \ge 0
\]
as required.

In summary, checking that an instance of $\xi$ is easily solvable can be done in time polynomial in $\log(1/\ve)$ using, for example, Miller-Rabin Primality Test, the average number of loop iterations is in $O(\log(1/\ve))$, an instance of the norm equation when $\xi$ is prime can be solved in time that is on average is in $O(\log^d(1/\ve))$ for some positive $d$. We conclude that on average the algorithm runs in time $O(\log^c(1/\ve))$ for some positive constant $c$.
\end{proof}

\begin{figure}[t]
\begin{algorithmic}[1]
\Require $\phi$ -- defines $R_z(\phi)X$, $\ve$ -- precision
\State $r\gets \sqrt{\p}, C \gets \sqrt{\p/(4r)}$
\State $m \gets \l\lceil \log_{\tau}(C\ve r)\r\rceil+1$
\State Find $k$ such that $\th = \phi/2+\pi/2-\pi k/5 \in [0,\pi/5]$
\State \emph{not-found} $\gets$ true, $u\gets0,v\gets0$
\While{\emph{not-found}}
  \State $u_0 \gets \text{RANDOM-SAMPLE}(\theta,\ve,r)$ \Comment See
  Figure~\ref{fig:appr:sampling}
  \State $\xi \gets \p^{2m} - \t \l|u_0\r|^2$
  \State $fl \gets \text{EASY-FACTOR}(\xi)$
  \If{$\text{EASY-SOLVABLE}(fl)$}
    \State \emph{not-found} $\gets$ false
    \State $v \gets \w^k \t^m u_0$
    \State $u \gets \t^m\text{SOLVE-NORM-EQUATION}(\xi)$
  \EndIf
\EndWhile
\State $C \gets \text{EXACT-SYNTHESIZE}(U[u,v,0])$
\Ensure Circuit $C$ such that $d(C,R_z(\phi)X)\le\ve$
\end{algorithmic}
\caption{\label{fig:rzxappr:algorithm} The algorithm for approximating $R_z(\phi)X$ by an $\FT$-circuit with $O(\log(1/\ve))$ gates and precision at most $\ve$. Runtime is probabilistic polynomial as a function of $\log(1/\ve)$. }
\end{figure}

The algorithm for approximating $R_z(\phi)X$~(Figure~\ref{fig:rzxappr:algorithm}) can now be easily constructed based on ideas discussed above. First we simplify the expression for the distance
\[
 d(U[u,v,0],R_z(\phi)X) = \sqrt{1-\sqrt{\t}\l|\re(v e^{-i(\phi/2+\pi/2)})\r|}
\]
and notice that in this case it depends only on the bottom left entry of the unitary $U[u,v,0]$. Now $u$ and $v$ have opposite roles in comparison to the algorithm for approximating $R_z(\phi)$. Again, to get a better constant factor in front of $\log(1/\ve)$ in the circuit size, we randomly pick $v_0$ such that
$d(U[u,\p^m v_0,0],R_z(\phi)X)\le\ve$. We use procedure RANDOM-SAMPLE to generate random $v_0$. When calling the procedure, we set the third input parameter $r$ to $\sqrt{\p}$ to take into account that bottom left entries of exact-unitaries are rescaled by factor $\sqrt{\t}$. Once we picked $v_0$ we check that there exist an exact unitary with bottom right entry $v=\t^m v_0 \sqrt{\t}$. In other words we solve norm equation
\[
 |u|^2 = \xi \text{ for }\xi=1 - \t|v|^2.
\]
The necessary condition $\gc{\xi}\ge0$ is always satisfied because $1+\p|\gc{v}|^2$ is always positive. Once we find an ``easy" instance of the norm equation we solve it and construct an exact unitary that gives the desired approximation. Our result regarding the approximation algorithm for $R_z(\phi)X$ is summarized by the following theorem.

\begin{thm}
Approximation algorithm~(Figure~\ref{fig:rzxappr:algorithm}) outputs a $\FT$-circuit $C$ of length $O(\log(1/\ve))$ such that $d(C,R_z(\phi)X)\le\ve.$ On average the algorithm runtime is in $O(\log^c(1/\ve))$ if Conjecture \ref{cnj:appr:distr} is true.
\end{thm}

The proof is completely analogous to the proof of Theorem~\ref{thm:appr:main} and we do not present it here.

\section{Experimental results}

In this section we evaluate the approximation quality of our algorithm as a function of $\sgm$-circuit size (depth) and the algorithm runtime. 
We not only confirm the results established in previous sections, but also show that constants hidden in the big-O notation are quite reasonable, making our algorithm useful in practice. 

\begin{table*}[tb]
  \centering
  \bgroup
\def\arraystretch{1.5}%
\pgfplotstabletypeset[
	column type=@{}l@{},
	every head row/.style={
		before row={%
			\hline
			& \multicolumn{3}{c|}{Unitaries} & \multicolumn{3}{c|}{Precisions} & \multicolumn{1}{l|}{Runs per}\\
		},
		after row=\hline,
	},
	every last row/.style={after row=\hline},
	every even row/.style={before row={\rowcolor[gray]{0.9}}},
	columns/name/.style       ={column name=~Name,column type=|@{}l@{}},
	columns/formula/.style  ={column name=~formula,column type=|@{}l@{}},
	columns/nmin/.style={column name=$k_{min}\quad$},
	columns/nmax/.style  ={column name=$k_{max}\quad$},
	columns/eformula/.style  ={column name=~formula$\quad$,column type=|@{}l@{}},
	columns/enmin/.style={column name=$k_{min}\quad$},
	columns/enmax/.style  ={column name=$k_{max}\quad$},
	columns/runs/.style={column name=~unitary,column type=|@{}l@{}|},
	col sep=&,row sep=\\,
	string type,
]{
name       & formula                                       & nmin &  nmax           & eformula  & enmin & enmax & runs \\
~U-RZ-SMALL$\quad$ &~$R_z(\frac{2\pi k}{10^{3}})$      & 1    & $1\cdot10^{3}$  & ~$10^{-k}$ & 2     & 14    &~1    \\
~U-RZ-BIG   &~$R_z(\frac{2\pi k}{4\cdot10^{3}})$       & 1    & $4\cdot10^{3}$  & ~$10^{-k}$ & 2     & 30    &~1    \\
~RZ-HIGH    &~$R_z(\frac{\pi }{2^k})$                  & 2    & $6\cdot10^{1}$  & ~$10^{-k}$ & 2     & 100   &~10   \\
~U-RZ-LOW   &~$R_z(\frac{2\pi k}{10^{3}})$             & 1    & $1\cdot10^{3}$  & ~$10^{-k/8}$&8     & 12    &~1    \\
~U-RZ       &~$R_z(\frac{\pi k}{10^{4}})$              & 1    & $1\cdot10^{4}$  & ~---       & ---   & ---   &~---  \\
~U-RZX-BIG  &~$R_z(\frac{2\pi k}{4\cdot10^{3}})X\quad$ & 1    & $4\cdot10^{3}$  & ~$10^{-k}$ & 2     & 30    &~1    \\
~U-RZX-LOW  &~$R_z(\frac{2\pi k}{10^{3}})X$            & 1    & $1\cdot10^{3}$  & ~$10^{-k/8}$&8     & 12    &~1    \\
~U-RZX      &~$R_z(\frac{\pi k}{10^{4}})X$             & 1    & $1\cdot10^{4}$  & ~---       & ---   & ---   &~---  \\
~X-HIGH     &~X                                            & ---  & ---             & ~$10^{-k}$ & 2     & 100   &~500  \\
}
\egroup

  \caption{ \label{table:experiments}Sets of inputs used for the experiments. }
\end{table*}

We experiment over several input sets of input unitaries and precisions, as summarized in Table~\ref{table:experiments}. 
Each experiment is performed similarly.
First, we request an approximation of a set of unitaries for certain precisions. 
In some experiments, we run the algorithm for the same unitary and precision several times to see the influence of the probabilistic nature of the algorithm on the result. 
Next, we aggregate collected data for a given precision by taking the mean, min or max of the parameter of interest over the set of all unitaries considered in the experiment. 
We compare to a Brute Force Search algorithm, from which we request approximation of the set of unitaries, that outputs the best precision that can be achieved using at most $N$ $\sigma$ gates for each input angle. 
The largest $N$ for our database is $25$. In this case we aggregate collected data for a given $N$.

We implemented our algorithm using C++. There are two third-party libraries used: PARI/GP\cite{UGPARI} which provides a relative norm equation solver and primality test, and \textit{boost::multiprecision} which includes high-precision integer and floating-point types. All experimental results described in this section were obtained on a computer with Intel Core i7-2600 (3.40GHz) processor and 8 GB of RAM. Our implementation does not use any parallelism.

\subsection{Quality Evaluation}
We evaluate the approximation quality of our algorithm on four large sets of inputs. 
Two of them are used to evaluate the approximation quality of $R_z(\phi)$ rotations and the other two for $R_z(\phi)X$. 
For both rotation types, one set covers uniformly the range of angles $[0,2\pi]$ (U-RZ-BIG, U-RZX-BIG) and the other one includes rotations that are particularly important in applications. 
For $R_z(\phi)$ rotations, we study angles $\phi$ of the form $\frac{\pi}{2^n}$ used in the Quantum Fourier Transform (input set RZ-HIGH); 
for $R_z(\phi)X$, we look at the quality of the Pauli $X$ gate approximation (input set X-HIGH). 
The results are presented in Figure~\ref{fig:big-exp}. In addition to the average number of gates, we show minimal and maximal number of gates needed to achieve the required precision, 
which demonstrates the stability of the quality of our algorithm.

One of the baselines we compare the quality of our algorithm to is Brute Force Search. We built a database of optimal $\sgm$-circuits with up to 25 gates and used it to find optimal approximations of unitaries from datasets U-RZ and U-RZX. The highest average precision that we were able to achieve is around $10^{-2.5}.$ We evaluated our number theoretic algorithm on the same range of precisions; the results are presented in Figure~\ref{fig:bfs:runtime}.
For our algorithm, the average coefficient in front of $\log_{10}(1/\ve)$ is only 18\% larger than the average for the optimal approximations of $R_z(\phi)$ and 40\% larger for $R_z(\phi)X$. 

Figure~\ref{fig:exact-stats} shows the exponential scaling of the number of optimal circuits with a given number of $\sgm$ gates and confirms that the Brute Force Search becomes infeasible exponentially quickly. 
In summary, our algorithm finds circuits for unitaries $R_z(\phi)$ and $R_z(\phi)X$ exponentially faster than Brute Force Search and with very moderate overhead.

The previous state-of-the-art method for solving the unitary approximation problem for the Fibonacci braid basis in polynomial time is the Solovay-Kitaev algorithm. 
The Solovay-Kitaev algorithm can be applied to any gate set and does not take into account the number theoretic structure of the approximation problem. 
Here we provide a rough estimate of its performance when approximating using $\sgm$-circuits. 

We consider approximation using special unitaries in this case. The version of the Solovay-Kitaev algorithm described in \cite{DN} boosts the quality of the approximation provided by a fixed-size epsilon net. 
It is crucial for the overall estimate to find the quality provided by an epsilon net depending on the maximal size of the optimal circuit in it. 
We use the scaling of the size of our database~(Figure~\ref{fig:exact-stats}) of optimal circuits and a rough volume argument to get the estimate. 

Consider the problem of approximating states, which is the same as approximating single-qubit special unitaries. Our $\ve$ net should cover the Bloch sphere
with overall surface area $4\pi$. Each state will cover roughly an area of the sphere equal to $\pi\ve^2$. Therefore, for $n$ being the size of the longest circuit:
\[
 \pi \ve^2 10^{(0.275x+0.592)}/2 \simeq 4\pi.
\]
We divide the estimate for the number of unitaries by two to get the estimate for the number of states. 
This is because there are only up to global phase two distinct exact unitaries of the form $U[x,y,k]$ for given $x,y$ (for k=0 and k=1). Other values of $k$ can be reduced to 0 or 1 using the identity $\w^s U[x,y,k] = U[x\w^s,y\w^s,k+2s]$. 
We also assume that $U[x,y,k]$ and $U[x\w^s,y\w^s,k]$ have very similar cost. 

The database we are using in other experiments includes circuits with up to $25$ gates and requires around 5GB of RAM to be built. In our estimate for the performance of the Solovay-Kitaev algorithm we assume that with enough engineering effort one can build a database with up to $30$ gates. Our estimates result in the following:
\[
\begin{array}{c}
\log_{10}(1/\ve_n) \simeq 0.137n -0.155, \log_{10}(1/\ve_{30}) \simeq  3.97 \\
n(\log_{10}(1/\ve_n)) \simeq  7.27\ve + 1.127.
\end{array}
\]
The estimate is more optimistic in comparison to results of our Brute Force Search as now we consider the full special unitary group instead of its subsets $R_z(\phi)$ and $R_z(\phi)X$. 
With these numbers in hand, we use the analysis of the Solovay-Kitaev algorithm in~\cite{DN}. 

Figure~\ref{fig:SK} compares our estimates for the size of circuits produced by the Solovay-Kitaev algorithm and the sizes from running our algorithm. In particular, for precision $10^{-10}$ our algorithm produces twenty times smaller circuits and for precision $10^{-30}$, one thousand times smaller circuits, which is an expected difference between algorithms with scaling $O(\log^{3.97}(1/\ve))$ and $O(\log(1/\ve))$.

\begin{figure*}[tb]
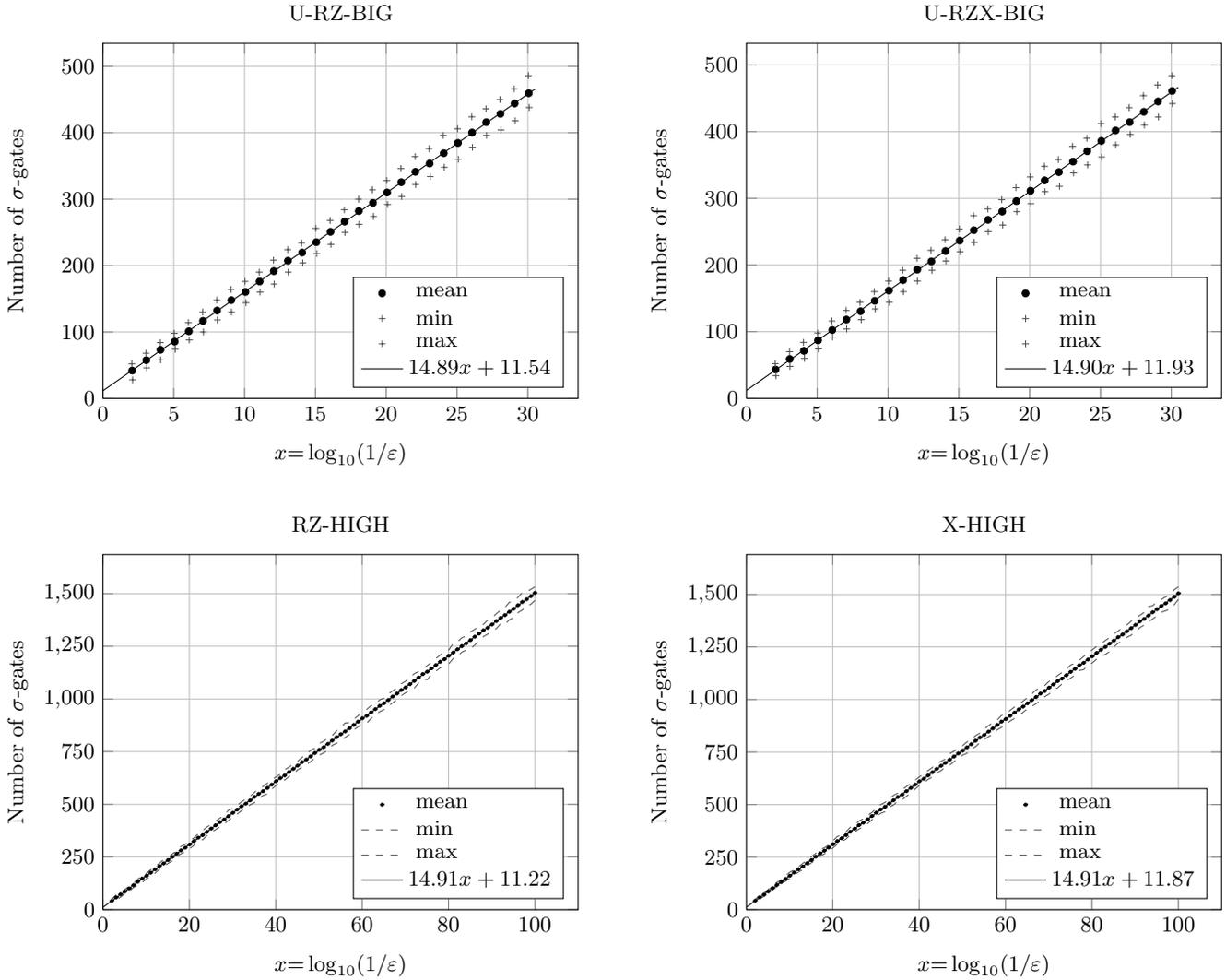

  \begin{subfigure}[b]{\columnwidth}
    \begin{tikzpicture}[trim axis left, trim axis right]
\pgfplotsset{
every axis plot post/.append style={
every mark/.append style={scale=0.7}}}
legend style={
}

\begin{axis}[grid=both, xmin = 0, ymin =0,title=U-RZ-BIG,
		xlabel=$x{=}\log_{10}(1/\ve)$,ylabel=Number of $\sigma$-gates, legend style={cells={anchor=west},
legend pos=south east,},y post scale=0.9, x post scale=1 ]
\selectcolormodel{gray}
\addplot[color=black,only marks,mark=*] coordinates {
(30.0627,459.513)
(29.0626,444.016)
(28.0628,428.387)
(27.0625,415.826)
(26.0627,400.29)
(25.0627,384.61)
(24.0626,369.133)
(23.063,353.622)
(22.0626,341.105)
(21.0628,325.497)
(20.0625,310.051)
(19.0625,294.413)
(18.0625,282.006)
(17.0627,266.227)
(16.0628,250.862)
(15.0628,235.257)
(14.0623,219.549)
(13.0626,207.24)
(12.0627,191.66)
(11.0624,176.026)
(10.0631,160.251)
(9.06275,147.856)
(8.06231,132.265)
(7.0624,116.678)
(6.06253,101.066)
(5.06289,85.489)
(4.06266,73.154)
(3.06246,57.548)
(2.06298,41.8635)
};
\addplot[color=gray!50!black!90,only marks,mark=+] coordinates {
(30.0904,438.)
(29.0988,418.)
(28.1044,404.)
(27.0825,396.)
(26.0862,378.)
(25.0934,360.)
(24.1002,348.)
(23.1102,334.)
(22.0837,322.)
(21.0883,304.)
(20.096,292.)
(19.1043,274.)
(18.081,262.)
(17.085,250.)
(16.0909,232.)
(15.099,218.)
(14.1074,204.)
(13.0817,190.)
(12.0871,172.)
(11.0931,160.)
(10.1001,144.)
(9.07961,130.)
(8.08368,118.)
(7.08877,100.)
(6.09646,88.)
(5.10411,74.)
(4.08085,58.)
(3.08586,46.)
(2.09113,28.)
};
\addplot[color=gray!50!black!90,only marks,mark=+] coordinates {
(30.0372,486.)
(29.0312,466.)
(28.0266,450.)
(27.0442,436.)
(26.0406,424.)
(25.0357,406.)
(24.0305,396.)
(23.0238,376.)
(22.043,364.)
(21.0385,346.)
(20.0345,328.)
(19.0269,314.)
(18.0452,300.)
(17.0416,284.)
(16.0371,268.)
(15.0315,256.)
(14.0246,234.)
(13.0447,224.)
(12.0403,208.)
(11.0351,190.)
(10.0305,176.)
(9.04662,164.)
(8.0428,148.)
(7.03878,130.)
(6.03441,114.)
(5.02759,98.)
(4.04527,84.)
(3.04183,68.)
(2.03649,52.)
};

\addplot[domain=0:30.5] {14.891*x+ 11.540};

\legend{~mean,~min,~max,$14.89x+11.54$}
\end{axis}%
\end{tikzpicture}%
  \end{subfigure}%
  \hfill
  \begin{subfigure}[b]{\columnwidth}
    \begin{tikzpicture}[trim axis left, trim axis right]
\pgfplotsset{
every axis plot post/.append style={
every mark/.append style={scale=0.7}}}
legend style={
}

\begin{axis}[grid=both, xmin = 0, ymin =0,title=U-RZX-BIG,
		xlabel=$x{=}\log_{10}(1/\ve)$,ylabel=Number of $\sigma$-gates, legend style={cells={anchor=west},
legend pos=south east,},y post scale=0.9, x post scale=1 ]
\selectcolormodel{gray}
\addplot[color=black,only marks,mark=*] coordinates
{(30.062488567570424768,461.03229036295369212)
(29.062401280020362107,445.23754693366708385)
(28.062862385522033771,429.74750000000000000)
(27.062748382625796978,414.17550000000000000)
(26.062602074094598496,401.73400000000000000)
(25.062599025556542533,386.07500000000000000)
(24.062622706224710883,370.62550000000000000)
(23.062776073996160463,355.04831038798498123)
(22.063008944948731126,339.30150000000000000)
(21.062570384809674350,326.90312891113892365)
(20.062429774025907047,311.37171464330413016)
(19.062756734049284857,295.82550000000000000)
(18.062599140106397411,280.16550000000000000)
(17.062554463707885036,267.76800000000000000)
(16.062608194440868769,252.20200000000000000)
(15.062728392629764428,236.58022528160200250)
(14.062386593473815937,220.99724655819774718)
(13.063033059966777703,205.41677096370463079)
(12.062591581397071996,192.98500000000000000)
(11.062474231175719087,177.35000000000000000)
(10.062791524215622825,161.66650000000000000)
(9.0631625101317691682,146.29300000000000000)
(8.0631494297096457211,130.59750000000000000)
(7.0626015881572633307,118.12500000000000000)
(6.0626897513029504147,102.51650000000000000)
(5.0628266640488287821,87.013000000000000000)
(4.0629149638943502702,71.309000000000000000)
(3.0624645472518208962,59.026282853566958698)
(2.0626579522495666000,43.285106382978723404)
};

\addplot[color=gray!50!black!90,only marks,mark=+] coordinates 
{(30.084040111879792493,442.00000000000000000)
(29.089982494687670829,422.00000000000000000)
(28.096732383462441269,410.00000000000000000)
(27.104052013944180118,396.00000000000000000)
(26.081470261180789398,380.00000000000000000)
(25.086567107988122506,362.00000000000000000)
(24.091664763236456991,350.00000000000000000)
(23.099721079726839363,338.00000000000000000)
(22.108592955047225502,318.00000000000000000)
(21.083019941962011352,310.00000000000000000)
(20.088358610303134964,292.00000000000000000)
(19.094458749344818901,280.00000000000000000)
(18.101377283971484620,260.00000000000000000)
(17.080041654901616135,250.00000000000000000)
(16.084781234119877033,234.00000000000000000)
(15.089868616170147622,220.00000000000000000)
(14.097529246585261750,206.00000000000000000)
(13.104553062700942900,192.00000000000000000)
(12.081240967091667641,176.00000000000000000)
(11.086573399019919425,160.00000000000000000)
(10.091882913638909694,144.00000000000000000)
(9.1000971434707447797,134.00000000000000000)
(8.1092244218235382556,118.00000000000000000)
(7.0826675823884881107,104.00000000000000000)
(6.0879291862326773212,92.000000000000000000)
(5.0941237586169871649,74.000000000000000000)
(4.1019196744582115621,60.000000000000000000)
(3.0801706029735918536,48.000000000000000000)
(2.0847344857615539665,34.000000000000000000)
};

\addplot[color=gray!50!black!90,only marks,mark=+] coordinates 
{(30.042091347844086223,484.00000000000000000)
(29.038235844587035302,470.00000000000000000)
(28.032586579620030887,454.00000000000000000)
(27.026507914694052466,436.00000000000000000)
(26.045025468629147307,422.00000000000000000)
(25.041459990635695236,412.00000000000000000)
(24.037256685316042781,390.00000000000000000)
(23.030193157310513089,378.00000000000000000)
(22.024366249354440995,358.00000000000000000)
(21.043432525435319270,348.00000000000000000)
(20.040404319134909068,332.00000000000000000)
(19.034513107225231497,316.00000000000000000)
(18.028912478356020804,298.00000000000000000)
(17.046575268898535362,284.00000000000000000)
(16.042776924654679302,274.00000000000000000)
(15.037888990213757391,254.00000000000000000)
(14.032343883606736206,238.00000000000000000)
(13.026822074734801276,222.00000000000000000)
(12.044914143949693078,210.00000000000000000)
(11.041203507833519641,192.00000000000000000)
(10.036099029171254226,176.00000000000000000)
(9.0309650547650774443,160.00000000000000000)
(8.0236767396566950083,144.00000000000000000)
(7.0436049159992613319,132.00000000000000000)
(6.0404788911632351086,116.00000000000000000)
(5.0352981124701656069,98.000000000000000000)
(4.0285967902253557090,84.000000000000000000)
(3.0467289570393519183,70.000000000000000000)
(2.0419798748723174012,52.000000000000000000)
};

\addplot[domain=0:30.5] {14.907655134244894489*x+11.935121360531369925};

\legend{~mean,~min,~max,$14.90x+11.93$}
\end{axis}%
\end{tikzpicture}%
  \end{subfigure}%
  \vspace{0.5cm}
   \begin{subfigure}[b]{\columnwidth}
    \input{plots/cost-RZ-HIGH.tex}
  \end{subfigure}%
  \hfill
  \begin{subfigure}[b]{\columnwidth}
    \input{plots/cost-X-HIGH.tex}
  \end{subfigure}%

\caption{ \label{fig:big-exp} Number of $\sigma$ gates needed to achieve the quality of approximation $\ve$ using the Number Theoretic Algorithm on different sets of inputs (see Table~\ref{table:experiments}). Includes approximation of $X$ gate and $R_z$ rotations used in the Quantum Fourier Transform.}
\end{figure*}

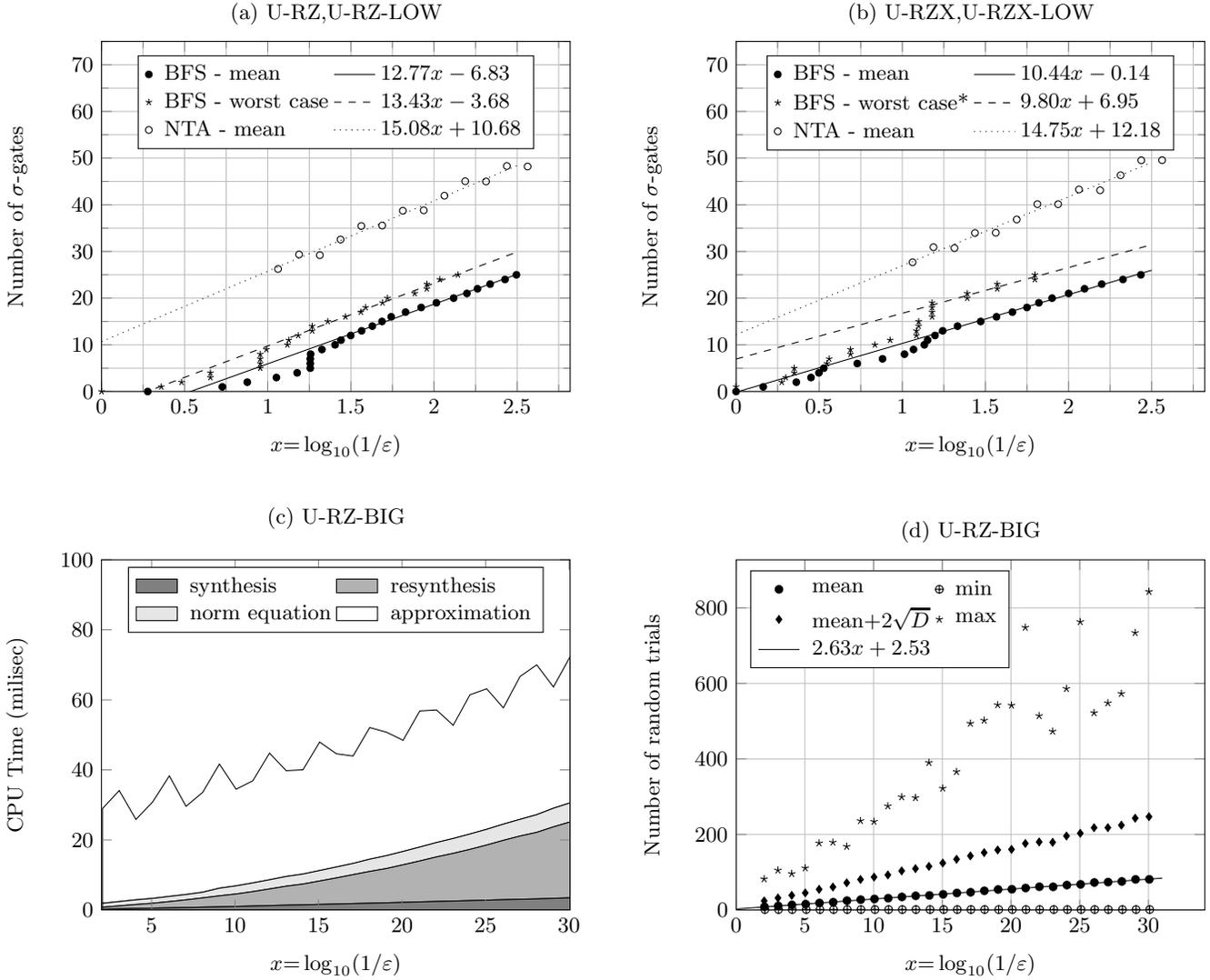
\begin{figure*}[tb]
  \begin{subfigure}[b]{\columnwidth}
    \caption{\label{fig:bfs:rz}U-RZ,U-RZ-LOW}
    \begin{tikzpicture}[trim axis left, trim axis right]
\pgfplotsset{
every axis plot post/.append style={
every mark/.append style={scale=0.7}}}
legend style={
}
\begin{axis}[grid=both, xmin = 0, ymin =0, ymax = 75,ytick ={0,10,...,80},xtick ={0,0.5,...,2.5},minor tick num=1,
		xlabel=$x{=}\log_{10}(1/\ve)$,ylabel=Number of $\sigma$-gates, legend style={cells={anchor=west},anchor=north,at={(0.5,0.96)} },legend columns=2,y post scale=0.9, x post scale=1 ]
\selectcolormodel{gray}

\addplot[color=black,only marks,mark=otimes*] coordinates {
(0.277866,0)
(0.726692,1)
(0.876889,2)
(1.05203,3)
(1.17669,4)
(1.25565,5)
(1.25565,6)
(1.25565,7)
(1.25767,8)
(1.32583,9)
(1.40469,10)
(1.44045,11)
(1.49888,12)
(1.56462,13)
(1.62905,14)
(1.6872,15)
(1.7427,16)
(1.82898,17)
(1.9229,18)
(2.01452,19)
(2.11809,20)
(2.19787,21)
(2.26174,22)
(2.33755,23)
(2.42717,24)
(2.49574,25)
};

\addplot[domain=0:2.5] {12.7725*x-6.83612};

\addplot[color=black,only marks,mark=star] coordinates
{
(0,0)
(0.35950263808196947349,1)
(0.4812997487535758630,2)
(0.6551525608110203359,3)
(0.6551525608110208263,4)
(0.9548417037615698998,5)
(0.9548417037615698998,6)
(0.9548417037615698998,7)
(0.9548417037615698998,8)
(0.9918536635194154688,9)
(1.1181206651647066435,10)
(1.1265607844860146328,11)
(1.1831193285293601300,12)
(1.2671299135925779519,13)
(1.2671299135925944523,14)
(1.3610855160334382269,15)
(1.466988200282372845,16)
(1.559843764111076170,17)
(1.588793851608131891,18)
(1.690244065540871723,19)
(1.720276256665779074,20)
(1.884113920772199260,21)
(1.956828531903444108,22)
(1.957559963532745469,23)
(2.035499022755510258,24)
(2.144310306503067592,25)
};

\addplot[domain=0:2.5,style=dashed] {13.43393331916659480*x-3.68190565865430403};

\addplot[color=black,only marks,mark=o] coordinates
{
(2.56292,48.198)
(2.43806,48.3)
(2.31286,44.994)
(2.18774,45.0401)
(2.06278,41.944)
(1.9382,38.818)
(1.81264,38.726)
(1.68854,35.552)
(1.56244,35.458)
(1.43781,32.562)
(1.31299,29.206)
(1.18811,29.35)
(1.06182,26.249)
};

\addplot[domain=0:2.5,style=dotted]{15.084378421573019416*x+10.684984399748293866};
\legend{~BFS - mean,$12.77x-6.83$,~BFS - worst case,$13.43x-3.68$,~NTA - mean,$15.08x+10.68$}
\end{axis}%
\end{tikzpicture}%
  \end{subfigure}%
  \hfill
  \begin{subfigure}[b]{\columnwidth}
    \caption{\label{fig:bfs:rzx}U-RZX,U-RZX-LOW}
    \begin{tikzpicture}[trim axis left, trim axis right]
\pgfplotsset{
every axis plot post/.append style={
every mark/.append style={scale=0.7}}}
legend style={
}
\begin{axis}[grid=both, xmin = 0, ymin =0, ymax = 75,ytick ={0,10,...,80},xtick ={0,0.5,...,2.5},minor tick num=1,
		xlabel=$x{=}\log_{10}(1/\ve)$,ylabel=Number of $\sigma$-gates, legend style={cells={anchor=west},anchor=north,at={(0.5,0.96)} },legend columns=2,y post scale=0.9, x post scale=1 ]
\selectcolormodel{gray}

\addplot[color=black,only marks,mark=otimes*] coordinates {
(0,0)
(0.16366633426181648273,1)
(0.3623873400429564474,2)
(0.4508888879486324032,3)
(0.4982695722112438251,4)
(0.5278121899246723158,5)
(0.7288260745815368192,6)
(0.8809103652766920183,7)
(1.0131650247703077354,8)
(1.0675672253224406988,9)
(1.1325469410820568644,10)
(1.1520589724948555460,11)
(1.1965469290809889017,12)
(1.2429170285367476790,13)
(1.3323205799128657668,14)
(1.471299728444090054,15)
(1.566437729615744636,16)
(1.662218983833746368,17)
(1.749732283588184623,18)
(1.821271289495724491,19)
(1.900869332567594658,20)
(2.000122547165817362,21)
(2.09618,22)
(2.19938,23)
(2.32723,24)
(2.43507,25)
};

\addplot[domain=0:2.5] {10.4472*x-0.142678};

\addplot[color=black,only marks,mark=star] coordinates
{
(0,0)
(0,1)
(0.27504376248488107012,2)
(0.29901933460630983540,3)
(0.34998172693671099503,4)
(0.34998172693671099503,5)
(0.5460894919835928113,6)
(0.5600977191850870530,7)
(0.6882128641940471040,8)
(0.6882128641940471040,9)
(0.8349446314687519609,10)
(0.9265457221854759442,11)
(1.0849197226052187224,12)
(1.0849197226052294223,13)
(1.0983363529623194063,14)
(1.0983363529623876603,15)
(1.1791195025209445678,16)
(1.1791195025209445678,17)
(1.1791195025216210733,18)
(1.1791195025216210733,19)
(1.390833581673793312,20)
(1.390833581673793312,21)
(1.571133857961602881,22)
(1.571133857972107040,23)
(1.797258104070965619,24)
(1.797258104085370978,25)
};

\addplot[domain=0:2.5,style=dashed]{9.80039714679857604*x+6.95256652012240278};

\addplot[color=black,only marks,mark=o] coordinates
{(2.5627677556581934293,49.556)
(2.4374871574245833849,49.527638190954773869)
(2.3131710834049531409,46.329648241206030151)
(2.1896743973214295739,43.108)
(2.0626489799070565663,43.276381909547738693)
(1.9370448214179818736,40.134)
(1.8124038890124095517,40.146)
(1.6882061783238350785,36.858)
(1.5623332506858437616,34.016)
(1.437189191733357885,33.95778894472361809)
(1.3137147327714768938,30.726)
(1.18761232357674519,30.916)
(1.0616063538590257598,27.686)
};

\addplot[domain=0:2.5,style=dotted]{14.758072613969320663*x+12.188521730552985744};
\legend{~BFS - mean,$10.44x-0.14$,~BFS - worst case*,$9.80x+6.95$,~NTA - mean,$14.75x+12.18$}
\end{axis}%
\end{tikzpicture}%
  \end{subfigure}%
  \vspace{0.5cm}
   \begin{subfigure}[b]{\columnwidth}
    \caption{\label{fig:runtime}U-RZ-BIG}
    \begin{tikzpicture}[trim axis left, trim axis right]
\pgfplotsset{
every axis plot post/.append style={
every mark/.append style={scale=0.7}}}
legend style={
}

\begin{axis}[grid=none, xmin = 2, ymin =0, stack plots=y, ymax = 100,
		area style,
		enlarge x limits=false,
		xlabel=$x{=}\log_{10}(1/\ve)$,ylabel=CPU Time (milisec), legend style={cells={anchor=west},anchor=north,at ={(0.5,0.98)}},legend columns=2,y post scale=0.9, x post scale=1 ]
\selectcolormodel{gray}

\addplot[color=black,no marks,mark=otimes*,fill=gray] coordinates
{(30.0627,3.45972)
(29.0626,3.32104)
(28.0628,3.15002)
(27.0625,3.0562)
(26.0627,2.92733)
(25.0627,2.79713)
(24.0626,2.66863)
(23.063,2.53743)
(22.0626,2.42545)
(21.0628,2.29896)
(20.0625,2.16898)
(19.0625,2.03254)
(18.0625,1.92272)
(17.0627,1.8097)
(16.0628,1.69002)
(15.0628,1.57718)
(14.0623,1.45259)
(13.0626,1.36648)
(12.0627,1.24356)
(11.0624,1.12298)
(10.0631,1.00339)
(9.06275,0.897352)
(8.06231,0.794107)
(7.0624,0.697853)
(6.06253,0.602941)
(5.06289,0.50831)
(4.06266,0.42578)
(3.06246,0.334249)
(2.06298,0.245326)
}\closedcycle;

\addplot[color=black,no marks,fill=gray!60] coordinates 
{(30.0627,21.6765)
(29.0626,20.4026)
(28.0628,19.0058)
(27.0625,18.0549)
(26.0627,16.9422)
(25.0627,15.7364)
(24.0626,14.6224)
(23.063,13.625)
(22.0626,12.718)
(21.0628,11.6866)
(20.0625,10.7289)
(19.0625,9.79214)
(18.0625,9.07582)
(17.0627,8.26754)
(16.0628,7.4526)
(15.0628,6.68643)
(14.0623,5.94232)
(13.0626,5.41658)
(12.0627,4.75569)
(11.0624,4.11213)
(10.0631,3.55077)
(9.06275,3.13125)
(8.06231,2.64696)
(7.0624,2.18483)
(6.06253,1.78263)
(5.06289,1.40368)
(4.06266,1.13301)
(3.06246,0.827832)
(2.06298,0.530662)
}\closedcycle;

\addplot[color=black,no marks,fill=gray!20] coordinates 
{(30.0627,5.44265)
(29.0626,5.32407)
(28.0628,5.05159)
(27.0625,4.76869)
(26.0627,4.64019)
(25.0627,4.49727)
(24.0626,4.39962)
(23.063,4.24992)
(22.0626,4.11305)
(21.0628,3.98198)
(20.0625,3.8089)
(19.0625,3.68047)
(18.0625,3.49528)
(17.0627,3.20767)
(16.0628,3.1209)
(15.0628,2.9901)
(14.0623,2.85793)
(13.0626,2.77576)
(12.0627,2.61839)
(11.0624,2.48222)
(10.0631,2.32417)
(9.06275,2.1869)
(8.06231,1.62029)
(7.0624,1.52574)
(6.06253,1.47291)
(5.06289,1.38544)
(4.06266,1.33779)
(3.06246,1.23102)
(2.06298,1.1238)
}\closedcycle;

\addplot[color=black,no marks,mark=+] coordinates 
{(30.0627,41.7566)
(29.0626,34.634)
(28.0628,42.7965)
(27.0625,40.7951)
(26.0627,33.2063)
(25.0627,40.1413)
(24.0626,39.7509)
(23.063,32.3316)
(22.0626,37.8523)
(21.0628,38.8708)
(20.0625,31.7477)
(19.0625,35.2678)
(18.0625,37.6021)
(17.0627,30.6672)
(16.0628,32.3487)
(15.0628,36.7049)
(14.0623,29.7654)
(13.0626,30.2113)
(12.0627,36.1612)
(11.0624,29.1611)
(10.0631,27.6009)
(9.06275,35.457)
(8.06231,28.5268)
(7.0624,25.1809)
(6.06253,34.4251)
(5.06289,27.5468)
(4.06266,22.9505)
(3.06246,31.6988)
(2.06298,27.0707)
}\closedcycle;


\legend{~synthesis,~resynthesis,~norm equation,~approximation}
\end{axis}%
\end{tikzpicture}%
  \end{subfigure}%
  \hfill
  \begin{subfigure}[b]{\columnwidth}
    \caption{\label{fig:trials}U-RZ-BIG}
    \begin{tikzpicture}[trim axis left, trim axis right]
\pgfplotsset{
every axis plot post/.append style={
every mark/.append style={scale=0.8}}}
legend style={
}
\begin{axis}[grid=both, xmin = 0, ymin =0,
		xlabel=$x{=}\log_{10}(1/\ve)$,ylabel=Number of random trials, legend style={cells={anchor=west},legend pos=north west},legend columns=2,y post scale=0.9, x post scale=1 ]
\selectcolormodel{gray}

\addplot[color=black,only marks,mark=*] coordinates
{(30.0627,81.051)
(29.0626,80.968)
(28.0628,75.491)
(27.0625,73.7875)
(26.0627,72.7325)
(25.0627,67.77)
(24.0626,65.8497)
(23.063,61.2198)
(22.0626,61.4665)
(21.0628,58.12)
(20.0625,54.855)
(19.0625,54.2455)
(18.0625,50.7881)
(17.0627,46.7855)
(16.0628,44.875)
(15.0628,41.3303)
(14.0623,38.469)
(13.0626,37.4315)
(12.0627,34.5223)
(11.0624,31.1493)
(10.0631,29.1113)
(9.06275,27.1848)
(8.06231,24.6985)
(7.0624,20.4698)
(6.06253,18.3707)
(5.06289,15.369)
(4.06266,13.362)
(3.06246,10.7068)
(2.06298,8.317)
};

\addplot[color=black,only marks,mark=oplus] coordinates 
{(30.0904,1.)
(29.0988,1.)
(28.1044,1.)
(27.0825,1.)
(26.0862,1.)
(25.0934,1.)
(24.1002,1.)
(23.1102,1.)
(22.0837,1.)
(21.0883,1.)
(20.096,1.)
(19.1043,1.)
(18.081,1.)
(17.085,1.)
(16.0909,1.)
(15.099,1.)
(14.1074,1.)
(13.0817,1.)
(12.0871,1.)
(11.0931,1.)
(10.1001,1.)
(9.07961,1.)
(8.08368,1.)
(7.08877,1.)
(6.09646,1.)
(5.10411,1.)
(4.08085,1.)
(3.08586,1.)
(2.09113,1.)
};

\addplot[color=black,only marks,mark=diamond*] coordinates 
{(30.0408,246.966)
(29.0357,242.936)
(28.0305,224.411)
(27.0463,217.947)
(26.0432,217.789)
(25.0391,203.071)
(24.0342,195.849)
(23.0279,179.04)
(22.046,180.091)
(21.0424,176.096)
(20.0374,160.365)
(19.0316,158.868)
(18.0479,151.886)
(17.0447,143.166)
(16.0405,134.375)
(15.0355,124.3)
(14.0292,115.4)
(13.0469,109.43)
(12.0431,103.445)
(11.0391,92.7004)
(10.034,87.0439)
(9.04883,80.6738)
(8.0455,71.9326)
(7.04145,60.7248)
(6.0373,54.5201)
(5.03179,45.0678)
(4.04788,38.6688)
(3.04433,31.7478)
(2.04031,23.9506)
};

\addplot[color=black,only marks,mark=star] coordinates 
{(30.0372,843.)
(29.0312,734.)
(28.0266,573.)
(27.0442,548.)
(26.0406,522.)
(25.0357,763.)
(24.0305,586.)
(23.0238,473.)
(22.043,514.)
(21.0385,748.)
(20.0345,542.)
(19.0269,543.)
(18.0452,502.)
(17.0416,494.)
(16.0371,366.)
(15.0315,322.)
(14.0246,390.)
(13.0447,297.)
(12.0403,299.)
(11.0351,275.)
(10.0305,234.)
(9.04662,236.)
(8.0428,168.)
(7.03878,179.)
(6.03441,177.)
(5.02759,111.)
(4.04527,96.)
(3.04183,105.)
(2.03649,82.)
};

\addplot[domain=0:31]{2.6342102859877107486*x+2.5323373864146127733};

\legend{~mean,~min,~mean+2$\sqrt{D}$,~max,~$2.63x+2.53$}
\end{axis}%
\end{tikzpicture}%
  \end{subfigure}%
\caption{\label{fig:bfs:runtime}(a),(b) comparison of the number of $\sigma$ gates needed to achieve the quality of approximation $\ve$ using the Number Theoretic Algorithm (NTA) and using Brute Force Search (BFS). Input sets U-RZ-LOW and U-RZX-LOW were used for the NTA and U-RZ, U-RZX for BFS; (c) average runtime of parts of the NTA, using U-RZ-BIG input set; (d) number of trials performed in the main loop of the NTA before an ``easy'' instance was found, using U-RZ-BIG input set. See Table~\ref{table:experiments} for input set descriptions. }
\end{figure*}

\subsection{Performance evaluation}

Our experiments confirm that the algorithm described in the paper has a probabilistic polynomial runtime. In addition, constants and the power of the polynomial are such that our algorithm is practically useful. 
In Figure~\ref{fig:runtime} we show the runtime of different parts of our algorithm. 
The approximation part corresponds to the runtime of the algorithm approximating unitaries $R_z(\phi)$ by exact unitaries (Figure~~\ref{fig:rzappr:algorithm}), excluding time needed to solve the norm equation; the synthesis part corresponds to the runtime of the exact synthesis algorithm that produces an $\FT$-circuit; the resynthesis part corresponds to the runtime of peephole optimization performed on $\sgm$-circuits obtained from $\FT$-circuits using identities from Section~\ref{sec:exact}. 

We separately show the runtime of the relative norm equation solver because for our implementation we used a generic solver which is a part of the PARI/GP library (function \textit{rnfisnorm}~\cite{UGPARI}). 
Since the library documentation does not describe the function performance in detail, we performed an evaluation to confirm that it is polynomial time on ``easy'' instances of the problem. 
We also rely on another PARI/GP function \textit{ispseudoprime} to perform a primality test. It is a combination of several probabilistic polynomial time primality tests~\cite{UGPARI}. The runtime of the primality test is included in the approximation part of the figure. 

Our claim about the approximation algorithm runtime depends on Conjecture~\ref{cnj:appr:distr}; Figure~\ref{fig:trials} shows the number of trials performed in the main loop of the approximation algorithm before finding an easy instance. The scaling of the average number of trials shown in the figure supports the conjecture. To perform peephole optimization~\cite{pho} we used the database of optimal $\sgm$-circuits with up to 19 gates which has size 85.7 MB.

High-precision integer and floating-point data types are necessary to implement our algorithm. We use \textit{cpp\_int} and  \textit{cpp\_dec\_float} from \textit{ boost::multiprecision} library. The number of bits used by these types can be specified at compile time. This allows us to avoid dynamic memory allocation when performing arithmetic operations; much faster stack memory is used instead. For this reason, runtime scaling~(Figure~\ref{fig:runtime}) of our code is a function of the number of arithmetic operations, and not of the bit size of numbers used in the algorithm. 

In Figure~\ref{fig:types} we show how the runtime of our algorithm changes when using different arithmetic types on the same set of inputs. The first pair of types --- 512 bit integers and 200 decimal digits floating-point numbers --- is sufficient for precision up to $10^{-35}$, the second pair --- 1024 bits and 400 decimal digits --- for precision up to $10^{-70}$. Figure~\ref{fig:runtime} shows runtime scaling for different parts of our
algorithm in more detail when using the second pair of types. This shows that our algorithm is practical and can readily be used as a subroutine when compiling quantum algorithms with large numbers of different single-qubit operations.

\begin{figure*}[tb]
  \begin{subfigure}[b]{\columnwidth}
  \caption{\label{fig:types}U-RZ-SMALL}
    \begin{tikzpicture}[trim axis left, trim axis right]
\pgfplotsset{
every axis plot post/.append style={
every mark/.append style={scale=0.7}}}
legend style={
}

\begin{axis}[grid=both, xmin = 0, ymin =0,
		xlabel=$x{=}\log_{10}(1/\ve)$,ylabel=CPU Time (miliseconds), legend style={cells={anchor=west},
at={(0.96,0.55)},},y post scale=0.9, x post scale=1 ]
\selectcolormodel{gray}
\addplot[color=black,only marks,mark=*] coordinates
{(14.063893576989151158,17.213615948999999984)
(13.062816033396886568,16.529254426999999981)
(12.062852980756241204,17.481689174000000000)
(11.062722425021164458,14.151775667000000022)
(10.062513900422241418,12.924101679000000022)
(9.0623993704454137209,13.469837749498997977)
(8.0625541605487298698,10.915222522999999995)
(7.0628758758428541479,9.5310376470000000096)
(6.0623958060975505293,10.325897385771543084)
(5.0627609367226360759,8.5853505259999999942)
(4.0624825906077237648,7.3707525820000000042)
(3.0622213209692198454,8.0448451610000000007)
(2.0627306851033885418,6.7478975530000000121)
};

\addplot[color=black,only marks,mark=oplus] coordinates 
{(14.063296351304392673,44.099625640999999938)
(13.062350789261803510,43.957840232999999959)
(12.062829830525326902,50.174877263999999931)
(11.062632338652713005,40.927934241000000026)
(10.062623224562116674,38.600578121000000106)
(9.0626870553782398702,46.946854138276553051)
(8.0624126939088724283,37.859351029000000020)
(7.0632806099395449730,33.257141725000000026)
(6.0628532754268935930,43.941360659318637189)
(5.0634748866730421731,35.315015388000000005)
(4.0628024270481885975,29.447692707000000034)
(3.0625955103511773614,39.416734624999999964)
(2.0628855197721308081,33.360525200000000106)
};

\addplot[color=black,only marks,mark=diamond*] coordinates 
{(14.061903180917929922,179.79529795299999995)
(13.062815000330299009,181.27275711199999979)
(12.062450551078463630,224.63854672700000026)
(11.063383444199791841,178.83173334200000007)
(10.062636658384376184,170.93177202599999994)
(9.0628099090480654383,226.19630354208416805)
(8.0623815357199764483,178.58542714600000013)
(7.0621996836350811644,158.31965119499999981)
(6.0630149679425912670,223.61895856613226459)
(5.0632191343548662255,178.50161148300000020)
(4.0619805981298936126,147.06407036200000000)
(3.0630444299286922374,210.35754098799999985)
(2.0631774543558243281,178.58506528000000015)
};
\legend{~i512;f200,~i1024;f400,~i2048;f800}
\end{axis}%
\end{tikzpicture}%
  \end{subfigure}%
  \hfill
  \begin{subfigure}[b]{\columnwidth}
   \caption{\label{fig:exact-stats}}
    \begin{tikzpicture}[trim axis left, trim axis right]

\pgfplotsset{
every axis plot post/.append style={
every mark/.append style={scale=0.7}}}
legend style={
}
\begin{semilogyaxis}[grid=both, xmin = 0, ymin =0,xtick={0,4,...,26},minor tick num =3,
		xlabel=Optimal number of $\sigma$ gates,ylabel=Number of unitaries, legend style={cells={anchor=west},anchor=north,legend pos=north west},y post scale=0.9, x post scale=1 ]
	\selectcolormodel{gray}
\addplot[color=blue,only marks,mark=*] coordinates
{(0,1)
(1,4)
(2,12)
(3,25)
(4,48)
(5,94)
(6,176)
(7,330)
(8,624)
(9,1174)
(10,2210)
(11,4164)
(12,7842)
(13,14764)
(14,27806)
(15,52368)
(16,98616)
(17,185710)
(18,349736)
(19,658626)
(20,1240320)
(21,2335774)
(22,4398746)
(23,8283732)
(24,15599946)
(25,29377876)
};
\addplot[domain=0:25]{10^(0.592957 + 0.275047*x)};
\legend{BFS,$10^{(0.28x+0.59)}$}
\end{semilogyaxis}%
\end{tikzpicture}%
  \end{subfigure}%
    \vspace{0.5cm}
   \begin{subfigure}[b]{\columnwidth}
    \caption{}
    \begin{tikzpicture}[trim axis left, trim axis right]
\pgfplotsset{
every axis plot post/.append style={
every mark/.append style={scale=0.7}}}
legend style={
}

\begin{axis}[grid=both, xmin = 0, ymin =0,scaled y ticks = false,extra y ticks={160},extra y tick style={yticklabel pos=right,ytick pos=right},
      y tick label style={/pgf/number format/std,axis x line = top,
      /pgf/number format/1000 sep = \thinspace 
      },
		xlabel=$x{=}\log_{10}(1/\ve)$,ylabel=$\sigma$ gates, legend style={cells={anchor=west},
legend pos = north west,},y post scale=0.9, x post scale=1 ]
\selectcolormodel{gray}
\addplot[color=black,only marks,mark=*] coordinates
{(3.97064,30)
(5.20338,150)
(7.0525,750)
(9.82617,3750)
};
\addplot[domain=0:10,dashed] {14.891*x+ 11.540};
\legend{SKA,NTA~($14.89x+11.54$)}
\end{axis}%
\end{tikzpicture}%
  \end{subfigure}%
  \hfill
  \begin{subfigure}[b]{\columnwidth}
    \caption{}
    \begin{tikzpicture}[trim axis left, trim axis right]
\pgfplotsset{
every axis plot post/.append style={
every mark/.append style={scale=0.7}}}
legend style={
}

\begin{axis}[grid=both, xmin = 0, ymin =0,scaled y ticks = false,extra y ticks={458},extra y tick style={yticklabel pos=right,ytick pos=right},
      y tick label style={/pgf/number format/std,axis x line = top,
      /pgf/number format/1000 sep = \thinspace 
      },
		xlabel=$x{=}\log_{10}(1/\ve)$,ylabel=$\sigma$ gates, legend style={cells={anchor=west},
legend pos = north west,},y post scale=0.9, x post scale=1 ]
\selectcolormodel{gray}
\addplot[color=black,only marks,mark=*] coordinates
{(3.97064,30)
(5.20338,150)
(7.0525,750)
(9.82617,3750)
(13.9867,18750)
(20.2275,93750)
(29.5886,468750)
};
\addplot[domain=0:30,dashed] {14.891*x+ 11.540};
\legend{SKA,NTA~($14.89x+11.54$)}
\end{axis}%
\end{tikzpicture}%
  \end{subfigure}%
\caption{\label{fig:SK}(a) runtime of the Number Theoretic Algorithm when using different arithmetic data types. Type i$\mathbf{N}$ corresponds to a signed, unchecked boost::multiprecision::cpp\_int with MinDigits and MaxDigits set to $\mathbf{N}$; type f$\mathbf{N}$ corresponds to boost::multiprecision::cpp\_dec\_float with Digits10 set to $\mathbf{N}$; (b) number of distinct (up to a global phase) unitaries, such that their optimal implementation requires given number of $\sigma$ gates. (c),(d) comparison of the estimated size of circuits produced by the Solovay-Kitaev algorithm~(SKA) and the Number Theoretic Algorithm~(NTA).
}
\end{figure*}
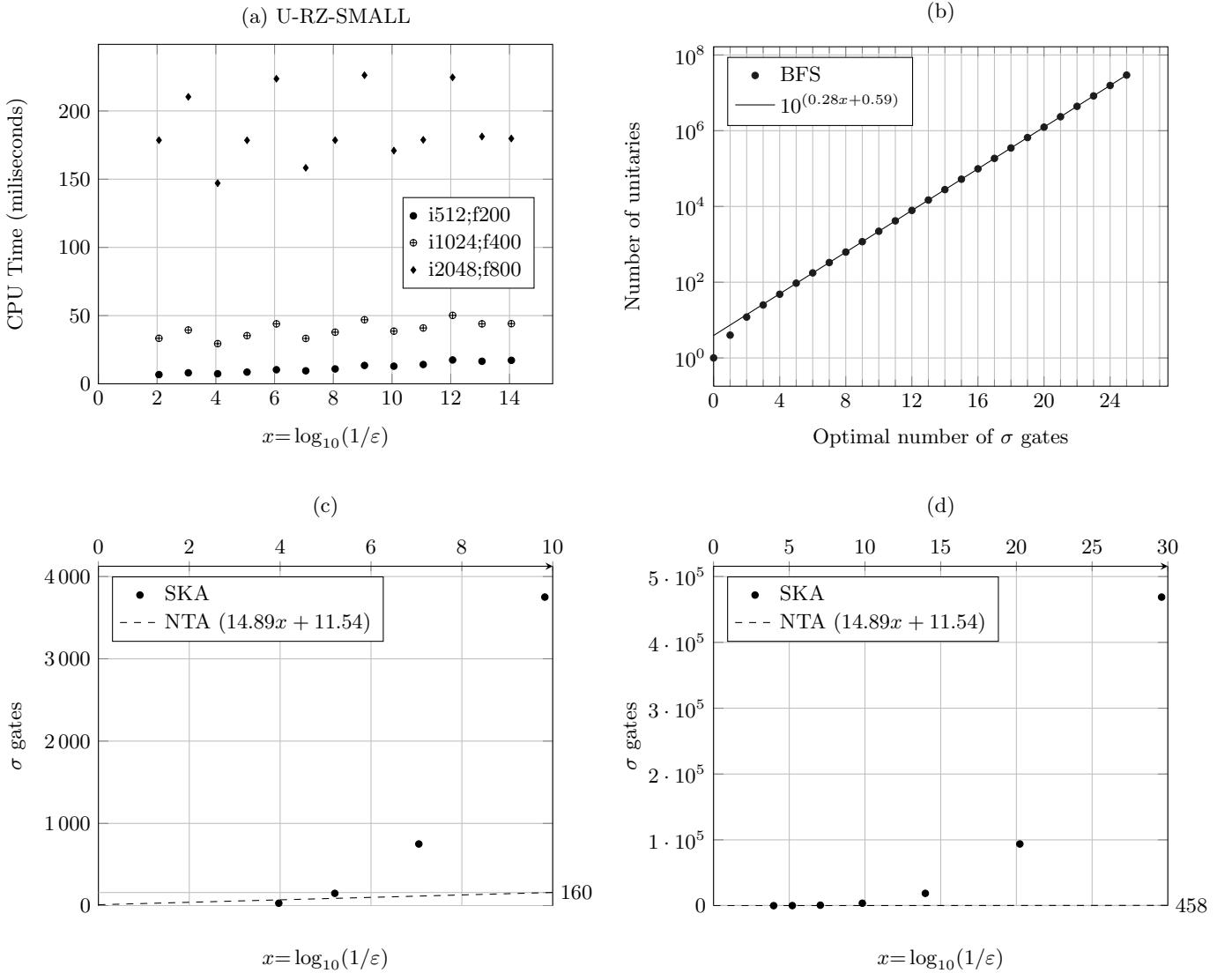

\section{Conclusions and Future Work}

We have considered the problem of optimal representation of single-qubit unitaries as braid patterns in the Fibonacci anyon basis, which is a promising non-Abelian anyon framework for topological computing. We have developed a set of principled solutions that enables the compilation of single-qubit unitaries into circuits of depth $O(\log(1/\ve))$ for an arbitrary target precision $\ve$;
the compiled circuits do not require ancillary qubits or pre-compiled resource states and are asyptotically depth-optimal.

The compiler runs on a classical computer in running time that is probabilistically polynomial in the bit size of $\ve$.
In practice, the runtime appears to scale slower than $\log^2(1/\ve)$ when $\ve$ tends to zero.
Compilation of an axial rotation to precision $10^{-30}$ takes less than 80 milliseconds on average on a regular classical desktop computer.
Consequently, our compiler improves on the most recent state-of-the-art solutions (c.f., \cite{HormoziEtAl}) in both an asymptotic and practical sense.

The availability of an asymptotically optimal compiler for single-qubit unitaries over the Fibonacci anyon basis lays a foundation for solving more challenging problems in topological quantum compilation, such as the problem of finding asymptotically optimal representations of multi-qubit unitary operations by circuits over that basis.
The fact that the multi-qubit Clifford group is not native to the Fibonacci anyon framework adds complexity to the challenge. It implies that either high-quality approximations of two-qubit Clifford gates must be designed or an entirely different set of entanglement gates must be considered, necessitating a translator between representations.
Thus, our future work will be multi-qubit compilation; we hope that the number theoretical methods described in this paper can be leveraged in the multi-qubit context.

Another direction is compilation of approximation circuits into ``weave" representations ~\cite{SimonFreedman,HormoziEtAl}.
A weave is a restricted braid where only one quasiparticle is allowed to move; a weave circuit is a circuit composed entirely of weaves.
It has been shown \cite{SimonFreedman,HormoziEtAl} that any braid pattern circuit can be approximated by a weave circuit.
Achieving an end-to-end compilation into weaves in probabilistically polynomial runtime with minimal overhead is an important direction of our future work.

\begin{acknowledgments}
We wish to thank Andreas Blass, Yuri Gurevich, Matthew Hastings, Martin Roeteller, Jon Yard and Dave Wecker for useful discussions.  VK wishes to thank Dr.~Cameron L. Stewart for helpful discussions regarding continued fractions.
\end{acknowledgments}

%

\appendix
\label{app:normeq}

\section{Background on Number Field Extensions}
\label{sec:exact:math}

We present the key mathematical concepts needed to understand the number theory involved in our algorithms.
A systematic discourse of underlying mathematics can be found in \cite{HCohen2000I,HCohen96}.

In this section, a ``field" means by default a \textit{number field} which is, by definition, a finite extension of the field of rational numbers $\Q$.

The particular fields we work with in this paper are the cyclotomic field $\Q(\omega)$, where $\omega=e^{\pi \, i/5}$ is the tenth primitive root of unity, and its maximum real subfield $\Q(\t)$, where $\t=(\sqrt{5}-1)/2$ is the inverse of the golden ratio. It is worth noting that there is some freedom in how we represent a field. 
For example, $\Q(\t)$ can be alternatively generated by the golden ratio $\phi=(\sqrt{5}+1)/2$ or by $\sqrt{5}$. 
In the same vein, $\Q(\omega)$ can be alternatively generated by $\theta=\w+\w^4$. In these terms it is easier to see that $\Q(\omega)$ is a quadratic extension of $\Q(\t)$; indeed, $\theta^2 = \t-2$. 
Similarly, $\Q(\t)$ is a quadratic extension of $\Q$; indeed, $\t$ is a root of the monic polynomial $T^2+T-1=0$.

\subsection{Galois extension and relative norm}

Let $L/K$ be an algebraic field extension of order $n$.

\begin{defin}
A $K$-linear map $f:L \rightarrow L$ is called a \textbf{$K$-automorphism} if it is bijective and preserves field multiplication, i.e., $\forall a,b \in L, f(a\,b) = f(a)\,f(b)$.
\end{defin}
Clearly a composition of two $K$-automorphisms is a $K$-automorphism and the inverse of a $K$-automorphism ia a $K$-automorphism.

\begin{defin}
Field extension $L/K$ of order $n$ is called normal or \textbf{Galois} if there exists exactly $n\,\,K$-automorphisms of $L$. The set of all $K$-automorphisms of $L$ forms a group under composition that is called the \textbf{Galois group of $L/K$}, denoted by $Gal(L/K)$.
\end{defin}

We note the following:
\begin{prop}\label{galois:fixed}
If $L/K$ is a Galois extension then $K$ coincides with the set of fixed points of the group $Gal(L/K)$.
\end{prop}

Even though the notion of \textit{relative norm} is usually defined in a more general setting, we focus on the particular definition best suited for the Galois extensions.
\begin{defin}
 Given a Galois extension $L/K$ the \textbf{relative norm} maps $N_{L/K}: L \rightarrow K$ and is defined as follows
  \begin{equation}
  N_{L/K}(\alpha)=\prod_{\sigma \in Gal(L/K)}{\sigma(\alpha)}.
  \end{equation}
\end{defin}

Correctness of this definition follows from proposition \ref{galois:fixed}: since $N_{L/K}(\alpha)$ is a fixed point of the Galois group, it belongs to $K$.
It is also easy to see that $\forall \alpha,\beta \in L, N_{L/K}(\alpha\, \beta)=N_{L/K}(\alpha)\, N_{L/K}(\beta)$.

Elements $\{ \sigma(\alpha) | \sigma \in Gal(L/K) \}$ are commonly called (Galois) conjugates of the element $\alpha \in L$. Thus, the relative norm of an $\alpha \in L$ is the product of all its Galois conjugates.
For an element $\beta \in K \subset L$, all the Galois conjugates are the same and coincide with $\beta$. Therefore $N_{L/K}(\beta)=\beta^n$, where $n$ is the order of the extension (and thus the size of the Galois group), and also for $\beta \in K \subset L$ and $\alpha \in L, \,\, N_{L/K}(\beta\,\alpha)=\beta^n\,\alpha$.

\subsection{Galois group and relative norm in $\Ri$}

Let $\w$ be the tenth root of unity as defined in Sectoin \ref{preliminaries:omega}).
Consider the field extension $\Q(\w)/\Q$, which is, by definition, the smallest in $\mathbb{C}$ containing both the rational numbers $\Q$ and the $\w$.

Considering the fifth root of unity $\zeta_5 = e^{2\,i\,\pi/5} = \w^2$ , it is well known that that $\Q(\w)$ coincides with the $\Q(\zeta_5)$, which is recognized as one of the simpler \textit{cyclotomic fields} (see \cite{LWashington}).
The minimal monic polynomial of $\zeta_5$ over $\Q$ is $X^4-X^3+X^2-X+1$, which can be easily verified by factoring $X^{10}-1$ over $\Z$.
This means that $\Q(\w)$ is fourth-degree extension of $\Q$.
When viewed as a vector space over $\Q$, the $\Q(\w)$ is a vector space with basis $\{1,\w,\w^2,\w^3\}$ and $\w$ is an algebraic integer in $\Q(\w)$. It is known (see \cite{HCohen2000I})
that the ring of algebraic integers of $\Q(\w)$ coincides with $\Ri$ and that the latter also has $\{1,\w,\w^2,\w^3\}$ as its integral basis.


Let us view complex conjugation as a field automorphism:
\begin{equation}
(\cdot)^*:\Q(\w) \rightarrow \Q(\w).
\end{equation}
By direct computation
\begin{equation}
\w^*=\w^{-1}=1-\w+\w^2-\w^3
\end{equation}
is linear with integer coefficients.
Therefore its restriction to $\Ri$ is a ring automorphism:
\begin{equation}
(\cdot)^*:\Ri \rightarrow \Ri.
\end{equation}

Consider the subring $\Z[\t] \subset \Ri$ which is the minimal subring of $\Ri$ containing $\t=\w^2-\w^3$.
Since $\Zt$ is populated by real-valued elements, it remains fixed under complex conjugation.
Consider $\Q(\t) \subset \Q(\w)$, the field of fractions of the ring $\Zt$.
(Conversely, $\Zt$ is the ring of integers of the field $\Q(\t)$.)
As per \cite{LWashington}:

(1) $\Q(\w)/\Q(\t)$ is a Galois extension with the two-element Galois group consisting of identity and the complex conjugation.

(2) Hence $\Q(\t) = \Q(\w) \cap \mathbb{R}$ is the maximum real subfield of $\Q(\w)$.

Similarly, $\Zt = \Ri \cap \mathbb{R}$.
Since the minimal polynomial of $\t$ is $X^2+X-1$, $\Q(\t)/\Q$ is a quadratic extension.
The ring $\Zt$  and its fraction field $\Q(\t)$ play an extremely important role in most of our constructions.

Consider the ladder of extensions $\Q(\w)\supset\Q(\t)\supset\Q$.
The ring automorphism already introduced in (\ref{exact:synthesis:bullet}) extends to field automorphism in a natural way:
\begin{equation}
(\cdot)^{\bullet}: \Q(\w) \rightarrow \Q(\w) ; \,
(\w)^{\bullet}=\w^3.
\end{equation}

The Galois group of the fourth-order extension $\Q(\w)/\Q$ happens to be the cyclic group $\Z_4$ generated by the $()^{\bullet}$,
in particular the complex conjugation is $()^{\bullet\,\bullet}$.
By direct computation,
\begin{equation}
\t^{\bullet} = -(\t+1),
\end{equation}
therefore the $()^{\bullet}$ automorphism has correctly defined restrictions to both $\Q(\t)$ and $\Zt$.
Since both are quadratic over rationals (resp., over rational integers) the restriction must be an order two automorphism, which is also seen directly since complex conjugation is the identity on $\Q(\t)$.

We now define the norm maps for all the above field extensions:
\begin{eqnarray}
N_{\t}:\Q(\t) &\rightarrow& \Q ; \,
N_{\t}(\xi) = \xi \, \xi^{\bullet} \\
N_i : \Q(\w) &\rightarrow& \Q(\t) ; \,
N_i(\eta) = \eta \, \eta^* \\
N: \Q(\w) &\rightarrow& \Q ; \,
N(\eta) = N_{\t}N_i(\eta)
\end{eqnarray}
The properties of the complex conjugation and the $()^{\bullet}$ automorphism imply that all the norm maps have correctly defined restrictions to the respective rings of integers $\Ri$ and $\Zt$.

We are now ready to introduce the \textit{Gauss complexity measure} on the ring $\Ri$:
\begin{equation}
G: \Z(\w) \rightarrow \Z ; \, G(\eta) = N_i(\eta)+N_i(\eta)^{\bullet}
\end{equation}
(see \cite{HLenstraJr} for the intuition behind this measure).
The correctness of this definition follows from the fact that $G(\eta)^{\bullet}=G(\eta)$ and therefore $G(\eta)$ is rational, but also since both $N_i$ and $(\cdot)^{\bullet}$ are integral in the integral bases, it has to be a rational integer.

\section{Background on Relative Norm Equation}
\label{sec:app:B}

We introduce just enough algebraic number theory to handle Thm \ref{solve:prime:xi} in a principled way.
The theorem offers a solvability condition (2) that is best derived from the theory of cyclotomic fields as presented in \cite{LWashington}.
It also offers a specific form for a solution (statement (b)) that can be derived from prime ideal decomposition algorithm, best described in \cite{JNeukirch}.

\subsection{Prime ideals in a Galois extension}

Below the term ``ring" will stand for a commutative ring with unity.
We start with common definitions (c.f., \cite{HCohen2000I}).
\begin{defin}
(1) An additive subgroup $I$ of a ring $R$ is called an \textbf{ideal} if it is closed under multiplication by any element of $r \in R$, i.e., $r \, I = \{r \, a | a \in I\} \subset I$.

(2) The \textbf{sum} $I + J$ of two ideals $I$ and $J$ of a ring $R$ consists of pairwise sums $a+b, a\in I, b\in J$.

(3) The \textbf{product} $I \, J$ of two ideals $I$ and $J$ of a ring $R$ consists of finite sums of pairwise products $ a\, b , a \in I, b \in J$.

\end{defin}

Both the sum and the product of two ideals of a ring is an ideal of that ring. It is also easy to see that (1) summation of ideals is associative and commutative, (2) multiplication of ideals is associative, commutative and distributive with respect to summation of ideals, (3) that the ideal generated by zero is the neutral element with respect to the addition of the ideals and (4) the entire ring $R$ is an ideal that is the neutral element with respect to multiplication of its ideals.

We call an ideal $\I \subset R$ a \textit{proper} ideal if it does not coincide with the entire ring $R$.
\begin{defin}
Given an element of a ring $R$,  $g \in R$, the  ideal $g R = \{g\,r | r \in R\}$ is the \textbf{principal} ideal generated by the element $g$.
\end{defin}

\begin{defin}
An ideal $I$ of a ring $R$ is called a \textbf{prime} ideal if it cannot be represented as a product of two or more proper ideals of the ring $R$.
\end{defin}

\begin{prop}
A principal ideal $g R$ of a ring $R$ is prime if and only if $g$ is a prime element of the ring $R$.
\end{prop}

Consider a Galois field extension $L/K$ and let ${\cal O}_{L}$and
${\cal O}_{K}$ be the corresponding rings of integers.
\begin{defin}\label{inert:split}
A prime ideal ${\cal {\cal {\cal P\subset}}}{\cal O}_{K}$ is

(1) \textbf{inert} in ${\cal O}_{L}$ if ${\cal PO}_{L}$ is a prime
ideal in ${\cal O}_{L}$,

(2) \textbf{fully ramified} in ${\cal O}_{L}$ if ${\cal PO}_{L}=({\cal Q})^{d}$
where ${\cal Q}$ is a prime ideal in ${\cal O}_{L}$,

(3) \textbf{fully split} in ${\cal O}_{L}$  if ${\cal P}{\cal O}_{L}=\prod_{j=1}^{d}{\cal Q}_{j}$
where ${\cal Q}_{j},j=1,...,d$ are pairwise distinct prime ideals in ${\cal O}_{L}$.
\end{defin}

It follows from field extension theory (c.f., \cite{HCohen2000I}) that the exponent
$d$ in clauses (2),(3) of Definition \ref{inert:split} is equal to the order of the Galois
group $Gal(L/K)$. 

\subsection{Split primes in the cyclotomic field $\Q(\w)$}

We apply these concepts to the number fields addressed in this paper.
We have a ladder of field extensions $\Q(\w)/\Q(\t)$ and $\Q(\t)/\Q$ and corresponding extensions of the rings of integers $\Ri/\Zt$, $\Zt/\Z$.
All the listed extensions are quadratic extensions and the \textit{absolute} extensions $\Q(\w)/\Q$  and $\Ri/\Z$ are order 4.
All the Galois groups of all extensions above are cyclic.

Recall that $\Q(\w)$ coincides with the \textit{cyclotomic} field $\Q(\zeta_5)$ where $\zeta_5$ is the fifth primitive root of unity.
It is easy to establish that an inert prime of the extension  $\Zt/\Z$ is a rational integer prime $p$ such that $p = \pm 2 \m 5$.
It is also easy to prove that given a rational integer prime $p$ such that $p = \pm 1 \m 5$, it is going to be fully split in the extension $\Zt/\Z$. That is, there exists a $\xi \in \Zt$ such that $N_{\t}( \xi)=N_{\t}( \xi^{\bullet})=p$.

\begin{prop} \label{appendix:prop:split}
In the above context, given $\xi>0$ and $p=N_{\t}( \xi)$ is prime, the relative norm equation $N_i(x)=\xi$ has a solution if and only if $p$ is fully split in the absolute extension $\Ri/\Z$.
\end{prop}
\begin{proof}
Indeed, if $N_i(x_1)=\xi$ and $x_2,x_3,x_4$ are $Gal(\Q(\w)/\Q)$ conjugates of $x_1$, then $N(x_1)=x_1\,x_2\,x_3\,x_4=p$, where $N=N_{\t} \, N_i$ is the absolute norm.
Therefore $p$ is fully split.

The converse is also true under the standing condition that $\xi > 0$.
Suppose $p=x_1\,x_2\,x_3\,x_4$ is fully split. By relabeling we can ensure that $x_2=x_1^{\bullet}$, $x_3=x_2^{\bullet}=x_1^*$, $x_4=x_3^{\bullet}=x_2^*$. Let $\eta=N_i(x_1)=x_1 \, x_3$, then $x_2\, x_4 = N_i(x_2)=N_i(x_1)^{\bullet} = \eta^{\bullet}$. Hence $p=\eta \, \eta^{\bullet} = N_{\t}(\eta)$. The pair of $\t$-conjugate solutions of $N_{\t}(\eta)=p$ is unique up to units of a certain form, specifically $N_{\t}( \xi)=N_{\t}(\eta)$ means that either $\xi = \pm \t^{2k}\,\eta$ or $\xi = \pm \t^{2k} \eta^{\bullet}$. However, both $\eta$ and $\eta^{\bullet}$ are square norms of complex entities and thus both are positive. Given $\xi > 0$, $\xi = + \t^{2k}\,\eta$ or $\xi = + \t^{2k} \eta^{\bullet}$. In the first case, $N_i(\t^k \, x_1)=\xi$ and in the second case $N_i(\t^k \, x_2)=\xi$.
\end{proof}

\subsection{Proof of Statement (2) of Theorem \ref{solve:prime:xi}}
We prove the statement by invoking Theorem 2.13 of \cite{LWashington}.

In the case of the cyclotomic field $\Q(\zeta_5)$, if $\xi \in \Zt$ happens to be an inert integer prime of $\Zt/\Z$ then it is also inert in $\Z[\zeta_5]/\Zt$.
Indeed, in this case $\xi$ is a rational integer prime and $\xi = \pm 2 \m 5$. The smallest power $f$ such that $\xi^f = 1 \m 5$ is $4$ and coincides with the order of the $\Z[\zeta_5]/\Z$ extension. Thus $\xi$ does not split.

Otherwise, for the prime $\xi$ either $p=N_{\t}(\xi)=5$ or $p=N_{\t}(\xi)=\pm 1 \m 5$.
The easy case of $p=5$ is covered by Observation \ref{normeq:norm5}.
In the remaining case $p$ must be fully split in $\Z[\zeta_5]/\Z$  as per Proposition \ref{appendix:prop:split}. As per the cited Theorem 2.13, this is equivalent to $p=1 \m 5$, which is precisely what is claimed in Statement (2) of Thm  \ref{solve:prime:xi}.


\subsection{Prime decomposition of ideal in Galois extension}

We now prove Statement (b) of Thm \ref{solve:prime:xi}.
This requires a textbook fact about \textit{fully split} ideals and an algorithm that deals with such ideals.

Let $L/K$ be a Galois field extension.
\begin{thm}\label{my:theorem:1}
(c.f., \cite{JNeukirch})
If a prime ideal ${\cal {\cal {\cal P\subset}}}{\cal O}_{K}$ is fully split
in ${\cal O}_{L}$then ${\cal P}$ is uniquely (up to relabeling)
represented as a product of distinct ideals ${\cal P}{\cal O}_{L}=\prod_{j=1}^{d}{\cal Q}_{j}$,
where ${\cal Q}_{j},j=1,...,d$ are prime in ${\cal O}_{L}$. The
Galois group $Gal(L/K)$ acts faithfully and transitively in this
set of prime ideals. In particular, $d$ is the order of the Galois
group.
\end{thm}
Note that, in case $\Q[\w]/\Q[\t]$, $d=2$ and the two ideals
occurring in the prime decomposition of a split prime ideal are complex
conjugates of each other.

\begin{defin}\label{primitive:element}
$\theta\in L$ is a \textbf{primitive element} of the extension
$L/K$ if $L=K(\theta)$.
\end{defin}
In our case the primitive element of choice for the $\Q(\w)/\Q(\t)$ extension is  $\theta=\w+\w^{4}=-1+2\,\w-\w^{2}+\w^{3}=\sqrt{\t-2}$.

\begin{defin}\label{exceptional:ideal}
For a primitive element $\theta$ of the extension $L/K$,
consider the conductor ideal ${\cal C}_{\theta}=\{y\in{\cal O}_{L}|y\:{\cal O}_{L}\subset{\cal O}_{K}[\theta]\}$.
A prime ideal ${\cal {\cal {\cal P\subset}}}{\cal O}_{K}$ is \textbf{exceptional
with respect to }$\theta$ when it is not relatively prime with the
${\cal C}_{\theta}$.
\end{defin}
For the above choice of primitive element $\theta$ in $\Q(\w)/\Q(\t)$,
${\cal C}_{\theta}=2\Ri$, so the only exceptional prime ideal is
the inert ideal $(2)$. There are no exceptional split prime ideals
in this case.

\begin{algorithm} \label{norm:neukurch:algo} (c.f., \cite{JNeukirch})
Let ${\cal {\cal {\cal P\subset}}}{\cal O}_{K}$ be a non-exceptional
prime ideal and let $H(X)\in{\cal O}_{K}[X]$ be the monic minimal
polynomial for the primitive element $\theta$. 
Let $F={\cal O}_{K}/{\cal P}$ be the corresponding residue field
and let $h(X)\in F[X]$ be the reduction of the $H(X)$ modulo ${\cal P}$.
Compute the irreducible factorization $h(X)=h_{1}(X)^{e_{1}}...h_{r}(X)^{e_{r}}$in
$F[X]$.
Then
\[
{\cal P}{\cal O}_{L}={\cal Q}_{1}^{e_{1}}...{\cal Q}_{r}^{e_{r}},
\]
where ${\cal Q}_{j}={\cal P}{\cal O}_{L}+h_{j}(\theta){\cal O}_{L},j=1,...,r$.
\end{algorithm}

\begin{corol}\label{special:case}
(Special Case.)
In the context of Thm \ref{my:theorem:1}, let ${\cal {\cal {\cal P\subset}}}{\cal O}_{K}$
be a prime ideal that is split in ${\cal O}_{L}$ and non-exceptional
with respect to a primitive element $\theta$ of $L/K$.
Then

(1) $F={\cal O}_{K}/{\cal P}$ is a splitting field of the residual
mimimal polynomial, i.e., $h(X)=(X-m_{1})...(X-m_{d})$ in $F[X]$;

(2) If additionally ${\cal O}_{L}$is a principal ideal domain then
we can effectively find such $s\in{\cal O}_{L}$that ${\cal P}{\cal O}_{L}=N_{L/K}(s){\cal O}_{L}$. 
That is, we can solve the relative norm equation in the group of
ideals.
\end{corol}
Indeed, (1) is a straightforward consequence of the algorithm. For
(2), consider the ideal ${\cal Q}_{1}={\cal P}{\cal O}_{L}+(\theta-m_{1}){\cal O}_{L}$.
Since ${\cal O}_{L}$ is a principal ideal domain, ${\cal Q}_{1}=s_{1}{\cal O}_{L}$ for
some $s_{1}\in {\cal O}_{L}$. Recall that the Galois group $Gal(L/K)$
acts faithfully and transitively on the set of ideals $\{{\cal Q}_{j}\}$,
let $\sigma_{j}\in Gal(L/K)$ be the element mapping ideal ${\cal Q}_{1}$
onto the ideal ${\cal Q}_{j},j=2,..,d$ and let $s_{j}=\sigma_{j}(s_{1}),j=2,...,d$.
Obviously for each $j$, ${\cal Q}_{j}=s_{j}{\cal O}_{L}$ and ${\cal P}{\cal O}_{L}=\prod_{j=1}^{d}{\cal Q}_{j}=(\prod_{j=1}^{d}s_{j}){\cal O}_{L}=N_{L/K}(s_{1}){\cal O}_{L}$.

As a concluding observation, let ${\cal P}{\cal O}_{L}$ be the principal
ideal $\xi{\cal O}_{L}$ for some $\xi\in{\cal O}_{K}$. Then the above
equality of principal ideals means that $\xi\sim N_{L/K}(s_{1})$, i.e., $\xi=u\, N_{L/K}(s_{1})$ for some unit $u$. Since $\xi,N_{L/K}(s_{1})$
are in ${\cal O}_{K}$, so is the unit $u$.


\subsection{Proof of Statement (b) of Theorem \ref{solve:prime:xi}}

For brevity, we leave out the easy case of $p=N_{\t}(\xi) = 5$.
The minimal polynomial for $\theta$ is $H(X)=X^{2}-(\t-2)$, and
its reduction $h(X)=X^{2}-((\t-2)\m p)$.
If conditions (1),(2) of the theorem hold, and in particular $p = 1 \m 5$, $h(X)$ is guaranteed to split in $\Z_p$ into $h(X)=(X-m)(X+m) \m p, m \in \Z$.

As per Algorithm \ref{norm:neukurch:algo}, the ideal $\xi\Ri$ is the product of $\xi\Ri+(m-\theta)\Ri$
and $\xi\Ri+(m+\theta)\Ri$. As per the discussion around the Corollary
\ref{special:case}, $\xi=u\, N_{i}(s)$, where $u$ is a unit in $\Zt$
and $s$ is (for example) a generator of the principal ideal $\xi\Ri+(m-\theta)\Ri$.
We can take $s=GCD_{\Ri}(\xi,m-\theta)$ as such a generator.
Note that the necessary conditions of solvability of (\ref{eq:rnorm}),
$\xi>0,\xi^{\bullet}>0$ in this specific context turn out to be sufficient.
Indeed, they imply that $u>0,u^{\bullet}>0$ which makes $u$ equal
to $\t^{2k}$ for some $k\in\Z$. Then $x=\tau^{k}\, s$ is the desired
solution of (\ref{eq:rnorm}). It follows that the other solution is $x^{*}=\tau^{k}\, s^{*}$.


\subsection{Proof of Theorem \ref{solve:general:xi}}


To prove the theorem, we first need the following:
\begin{lem}
Let
\begin{equation}
\xi=\eta^{2}\xi_{1}...\xi_{r},\eta,\xi_{j}\in\Zt,r\in\Z,r\geq0,j=1,...,r\label{appendix:th3factorization}
\end{equation}
be a factorization of an $\xi \in \Zt$, where $\xi_{1},...,\xi_{r}$ are $\Zt$-primes
such that $\xi_{1}\Zt,...,\xi_{r}\Zt$ are pairwise distinct prime
ideals.
If $\xi >0, \xi^{\bullet} > 0$ then there exists an equivalent factorization where $\xi_j >0, \xi_j^{\bullet} > 0$ for each $j=1,...,r$
\end{lem}

\begin{proof}
For any choice of $\xi_{1},...,\xi_{r}$, there is an even
number of these that are negative. Replacing each negative $\xi_{j}$
by $-\xi_{j}$ we get an equivalent factorization of $\xi$. Assuming
further that $\xi_{1},...,\xi_{r}$ are all positive, we note that
there is an even number of these that have negative adjoint. Now split
the set $\{\xi_{j}|\xi_{j}^{\bullet}<0\}$ arbitrarily into two halves
and replace each $\xi_{j}$ in the first half by $\t\,\xi_{j}$ (ensuring
$(\t\,\xi_{j})^{\bullet}>0$) then replace each $\xi_{j}$ in the
second half by $\t^{-1}\xi_{j}$(ensuring $(\t^{-1}\xi_{j})^{\bullet}>0$).
Then this set of replacements leads to an equivalent factorization of $\xi$.
\end{proof}

We are finally ready to prove Thm \ref{solve:general:xi}.
Given a factorization (\ref{appendix:th3factorization}) and assuming as per the lemma that $\xi_j >0, \xi_j^{\bullet} > 0$ for each $j=1,...,r$,
the theorem can be restated to claim that the equation $N_i(x)=\xi$ is solvable if and only if all the equations $N_i(y)=\xi_j$ are
individually solvable.

\begin{proof}
If the factor equations are individually solvable, then the $N_i(x)=\xi$ is obviously solvable.

Conversely, suppose (\ref{eq:rnorm}) is solvable for the square-free
$\xi$, i.e., $\xi=\xi_{1}...\xi_{r}=N_{i}(x),x\in\Ri$ and additionally
$\xi_{j}>0,\xi_{j}^{\bullet}>0$ as explained above. We would like
to show that each of the equations $N_{i}(x_{j})=\xi_{j}$ is individually
solvable, $j=1,...,r$. Consider the ideal $\xi_{j}\Ri+x\,\Ri$ and
its principal ideal representation $s_{j}\,\Ri$. The complex conjugation
of the ideal is $\xi_{j}\Ri+x^{*}\,\Ri$ and is equal to $s_{j}^{*}\,\Ri$.
Since $N_{i}(s_{j})=s_{j}s_{j}^{*}$ generates the product of the two
ideals, it divides $\xi_{j}^{2}$. Since both $N_{i}(s_{j})$ and $\xi_{j}^{2}$
are in $\Zt$, $N_{i}(s_{j})$ divides $\xi_{j}^{2}$ in $\Zt$.
Since $\xi_{j}$ is prime, $N_{i}(s_{j})$ is associate to either $\xi_{j}^{2}$
or to $\xi_{j}$. The former is impossible, since $N_{i}(s_{j})$
also divides $x\, x^{*}=\xi$ and $\xi$ is square free. Thus $u\, N_{i}(s_{j})=\xi_{j}$,
where $u$ is a unit. By already familiar argument, $u>0,u^{\bullet}>0$
and thus $u=t^{2}$ is a perfect square. Therefore $N_{i}(t\, s_{j})=\xi_{j}$.
\end{proof}

\end{document}